\documentclass[lettersize,journal]{IEEEtran}

\usepackage{amsmath,amsfonts}
\usepackage{algorithm}
\usepackage{algpseudocode}
\usepackage{array}
\usepackage[caption=false,font=normalsize,labelfont=sf,textfont=sf]{subfig}
\usepackage{textcomp}
\usepackage{stfloats}
\usepackage{url}
\usepackage{verbatim}
\usepackage{graphicx}
\usepackage{balance}
\usepackage{mathtools}
\usepackage{bm}
\usepackage{amssymb}
\usepackage{optidef}
\usepackage{enumerate}
\usepackage{epstopdf}
\usepackage{fancyhdr} 
\usepackage{indentfirst}
\usepackage{multirow}
\usepackage{threeparttable}
\usepackage{CJK}
\usepackage{makecell}
\usepackage{cases}
\usepackage{cite}
\usepackage{booktabs} 
\usepackage{color}
\usepackage[justification=centering]{caption}
\def\BibTeX{{\rm B\kern-.05em{\sc i\kern-.025em b}\kern-.08em
		T\kern-.1667em\lower.7ex\hbox{E}\kern-.125emX}}

\newtheorem{Proposition}{\it Proposition}[section]
\newtheorem{proof}{\it Proof}[section]
\newtheorem{Lemma}{\it Lemma}[section]

\newtheorem{remark}{Remark}[section]

\makeatletter
\newcommand{\Rmnum}[1]{\expandafter\@slowromancap\romannumeral #1@}
\makeatother

\begin{document}
	
	\title{\Huge Adaptive Source-Channel Coding for Bi-static Integrated Sensing and Semantic Communications }
	\author{
		\IEEEauthorblockN{Haotian Wang,  Dan Wang,  Xiaodong Xu, \IEEEmembership{Senior Member, IEEE},  Chuan Huang, \IEEEmembership{Member, IEEE}, Hao Chen, \IEEEmembership{Member, IEEE}, Nan Ma, \IEEEmembership{Member, IEEE} and Ping Zhang, \IEEEmembership{Fellow, IEEE}}
		\thanks{
			A conference version of this work \cite{wang2026integrated}  was accepted by IEEE WCNC 2026. 
			
			H. Wang, X. Xu, and N. Ma are with the State Key Laboratory of Networking and Switching Technology, Beijing University of Posts and Telecommunications, Beijing, China, 100876, and also with the Department of Broadband Communication, Pengcheng Laboratory, Shenzhen, China, 518055. Emails: wanght@bupt.edu.cn, xuxiaodong@bupt.edu.cn, manan@bupt.edu.cn.

			D. Wang and H. Chen are with the Department of  Broadband Communications, Pengcheng Laboratory, Shenzhen, China, 518055, and D. Wang is also with the Guangdong Provincial Key Laboratory of Future Networks of Intelligence, The Chinese University of Hong Kong, Shenzhen, 518172. Emails: wangd01@pcl.ac.cn, chenh03@pcl.ac.cn.
			
			Chuan Huang is with the Shenzhen Institute for Advanced Study, the University of Electronic Science and Technology of China, Shenzhen, 518110, China, and the Shenzhen Future Network of Intelligence Institute, Shenzhen, 518172, China (e-mail: huangch@uestc.edu.cn).
			
			P. Zhang is with the State Key Laboratory of Networking and Switching Technology, Beijing University of Posts and Telecommunications, Beijing, China, 100876. Email: pzhang@bupt.edu.cn.

		}
	}
	\maketitle
	\begin{abstract}
		Semantic communication (SemCom)  has emerged as a new paradigm to facilitate the performance of integrated sensing and communication  systems in 6G, due to its potential to enhance transmission efficiency by transmitting task-relevant semantic features rather than raw bits. However, most of the existing works mainly focus on sensing data compression to reduce the subsequent communication overheads, without considering the integrated transmission framework for both the SemCom and sensing tasks. 
		This paper proposes a sensing-aware adaptive source-channel coding (SA-ASCC) and beamforming design framework for bi-static integrated sensing and SemCom  (ISSC) systems by jointly optimizing the coding rate for SemCom task and the transmit beamforming for both the SemCom and sensing tasks. Specifically, an end-to-end semantic distortion function is approximated by deriving an upper bound composed of source and channel coding induced components, and then a hybrid Cram\'{e}r-Rao bound (HCRB)   is  derived for target position under imperfect time synchronization due to the transceiver deployed at different places in our considered bi-static ISSC system. To characterize the achievable region between SemCom and sensing performance, a distortion minimization problem is formulated by considering the HCRB threshold, channel uses, and  power budget, which is non-convex due to the coupled design variables and the mixed-integer program. Subsequently, an  alternating optimization (AO) algorithm is proposed to decompose this problem into the model selection and joint rate and beamforming optimization subproblems, which are solved by the exhaustive search method and the combination of successive convex approximation and fractional programming, respectively. 
		Finally, simulation results demonstrate that our proposed scheme outperforms the conventional  deep joint source-channel coding (DJSCC)-water filling (WF)-zero forcing (ZF) and BPG-WF-ZF  benchmarks.
	\end{abstract}
	\section{Introduction}
	The sixth-generation (6G) communication systems are expected to support various intelligent applications, such as autonomous driving, virtual/augmented reality (VR/AR), and smart manufacturing, which impose  stringent demands for  high-speed transmission,  ubiquitous connectivity, and high-precision sensing   in integrated sensing and communication (ISAC) systems \cite{ 9737357, 10680280}. As the  massive volume of raw sensing data is usually generated in these intelligent applications,   semantic communication (SemCom), a new paradigm  to extract and transmit the latent  features of sensing data via deep neural network (DNN) instead of raw  bit, has been proposed to reduce the intractable transmission  overheads in ISAC systems \cite{11320919}. However, the inherent functional and operational differences between sensing and communications pose new design challenges for the integrated transmission framework to support both the SemCom and sensing tasks.
	
	Generally, SemCom-related researches are  divided into two categories, i.e., separate source-channel coding (SSCC)\cite{6316136,10175391}, which performs source and channel coding independently, and joint source-channel coding (JSCC)\cite{10387242,10747747,8054694,9398576,10702555}, which directly encodes the source data into channel input symbols for transmissions. Theoretical analysis has shown that SSCC can achieve optimality under infinite blocklength case\cite{10387242}, while JSCC has been  proposed for its advantage in practical finite blocklength scenarios\cite{10747747}. Recently, DNN as a promising technology  has been widely adopted into SemCom systems to implement JSCC with practical data sources \cite{8054694}. For text transmission, the authors in \cite{9398576} proposed a transformer-based deep JSCC (DJSCC) scheme  in  SemCom system, which demonstrated superior performance in low signal-to-noise-ratio (SNR) region compared to the
	conventional communication. For image transmission, the authors in\cite{10702555} proposed a DJSCC-based multi-task  SemCom framework by utilizing the convolutional autoencoder to extract the feature of  image  for autonomous vehicles, which also showed efficient image reconstruction and classification performance.  However, the above  DJSCC-based schemes directly encode the semantic features  into analog symbols, which is incompatible with the digital communication systems in practice.

	To overcome the above challenges, recent studies have focused on incorporating the quantization and the constellation mapping modules into  DJSCC\cite{9998051,10495330}. The authors in \cite{9998051} proposed a DJSCC-Q scheme for single Gaussian channel scenario by  introducing a quantization layer to generate symbols from finite constellations for image transmission, which achieves nearly identical performance to traditional DJSCC. To adapt  different channel conditions, the authors in \cite{10495330 } proposed a variational autoencoder (VAE) based joint coding-modulation scheme, which learns the probability distribution of constellation points under  varied SNRs to guarantee the transmission robustness. To further avoid the training overhead due to varied tasks and channels,   an adaptive source and channel coding (ASCC) architecture driven by the channel state, which is composed of semantic source coding based on DNNs and digital channel coding, has recently been developed for digital SemCom systems \cite{10845799,li2025adaptivesourcechannelcodingsemantic,yuan2025adaptivesourcechannelcodingmultiuser}. The authors in \cite{10845799} proposed a source-channel rate adaptation scheme for image transmission over a single Gaussian channel, which jointly optimizes the source and channel coding rates to minimize the E2E distortion and effectively mitigate the cliff effect inherent in the conventional separation-based architecture. Under the similar architecture, the authors in  \cite{li2025adaptivesourcechannelcodingsemantic} proposed a joint optimization design of source-channel coding rate and power allocation over parallel Gaussian channels  to minimize the weighted average E2E distortion for both image reconstruction and classification task. Furthermore, the authors in \cite{yuan2025adaptivesourcechannelcodingmultiuser} proposed a multiuser semantic and data transmission framework that jointly optimizes the coding rates, power allocation, and transmit beamforming. However, in the emerging applications with challenging propagation environments and limited wireless resources, such as low-altitude UAV inspection\cite{11328937} and vehicular cooperative sensing\cite{11514424}, existing ASCC architectures designed for the above communication-only tasks cannot effectively address the joint requirements of semantic transmission and environmental sensing.

On the other hand, ISAC as one of the key technologies in 6G achieves both the sensing and communication functions by sharing the wireless resources and hardware platforms\cite{9540344,11030566}. Most of the studies have concentrated on mono-static ISAC system to investigate the integrated sensing and communication gains  \cite{9652071,8386661}, while it may cause strong self-interference (SI) due to the transceivers deployed at the same place.  Recently, most of the research has shifted the paradigm to bi-static cases, which deploy the transceiver  at different places to avoid  strong SI\cite{10281382, 11155148}. The authors in\cite{10281382} designed the transmit waveform in bi-static ISAC systems  by  minimizing the symbol error rate (SER) under the Cramér-Rao Bound (CRB) constraint on angle  estimation, which achieved better SER performance and improved robustness to nonlinear phase noise. To enhance sensing performance, the authors in \cite{11155148}  minimized the CRB on target position estimation under communication SINR and transmit power constraints by optimizing transmit beamforming in  bi-static ISAC systems. However, all the above mentioned bi-static related  works considered an ideal condition with perfect time synchronization (TS), which is impossible for  bi-static  systems in practice. To investigate the impact of imperfect TS on ISAC  performance, the authors in \cite{10380513}  characterized it by unknown clock offsets in networked ISAC systems, which prevent accurate delay estimation  of the cross-link echo over the base station (BS)-target-other BS path and only the target-reflected signal over the direct link can be exploited for joint detection, thereby degrading the target detection performance.  More specifically, the  authors in \cite{10684491} modeled the TS error as
a Gaussian random variable and derived the CRB on position
estimation with TS error, which showed that the boundary of
achievable performance region degrades with the increasing of
TS errors.
	
	Recently, the combination of SemCom into ISAC systems has  attracted extensive attention from both the academic and industrial spheres\cite{10750351,jia2024infrastructureassistedcollaborativeperceptionautomated,10417099}. The authors in \cite{10750351} proposed a semantic-based sensing data feedback scheme for aerial intelligent transportation scenarios, where a semantic database is utilized to extract high-level sensing features, thereby reducing feedback overhead and latency.  Similarly, the authors in \cite{jia2024infrastructureassistedcollaborativeperceptionautomated} proposed a four-step compression scheme for autonomous driving scenarios, which effectively reduces bandwidth requirements without compromising sensing accuracy. Furthermore, the authors in \cite{10417099} proposed an attention-based sensing feature prioritization scheme  in autonomous driving scenarios to prioritize compressed sensing data, which  enhances communication reliability.   
	However, most of the above mentioned  studies have concentrated on  sensing data compression to reduce the transmission overhead, without considering the design of  integrated transmission framework for jointly serving the SemCom and sensing tasks.

	 To address the above issues, this paper proposes a sensing-aware adaptive source-channel coding (SA-ASCC) scheme for bi-static integrated SemCom and sensing (ISSC) systems with imperfect TS, which  comprise one transmitter (TR), one target, and two receivers (REs). The TR transmits the  ISSC signal, which superimposes the SemCom data  with the sensing  data via beamforming technique for both the SemCom and sensing tasks, thereby inherently inducing a performance trade-off between these two tasks with shared transmission resources. As  SemCom typically performs more effectively in low-SNR region, a larger proportion of resources could be allocated for sensing, which motivates us to investigate  the potential  integrated performance gain in our considered bi-static ISSC system to obtain a larger achievable region between the SemCom and sensing tasks. Unlike the conventional ASCC schemes in \cite{10845799,li2025adaptivesourcechannelcodingsemantic,yuan2025adaptivesourcechannelcodingmultiuser}, where source-channel rate adaptation is only determined by communication channel conditions, the integration of sensing tasks introduces additional accuracy requirements that compete for limited resources in our considered system. Consequently, the optimal source-channel coding rates become jointly coupled with sensing requirements and channel conditions,  rendering existing  ASCC schemes inapplicable and necessitating a new SA-ASCC framework beyond existing communication-oriented designs. The main contributions  of this work are summarized as follows:
	
	\begin{itemize}
		\item
		First, we derive the ISSC signal models for transmission at the integrated  TR and reception at the sensing RE (SRE) and communication RE (CRE), respectively. Based on the derived signal models, we propose an integrated transmission framework for both the SemCom and sensing tasks, where the communication information undergoes semantic source encoding, channel encoding, and modulation, and then is  superimposed with the modulated  sensing probing data via beamforming technique.
		\item
		Then, an E2E semantic  distortion function is derived by a regression-based fitting method for any fixed source coding rate. To investigate the impact of TS errors on ISSC
		performance, we also derive the hybrid CRB (HCRB)
		for target position under imperfect TS  by modeling the TS error as a real-valued random variable with a known prior distribution.

		\item Finally, to characterize the achievable region between SemCom and sensing, an E2E distortion minimization problem is formulated by considering the  HCRB threshold, channel uses, and power budget, which is non-convex due to the coupled design variables and the mixed-integer program. To address this problem,  an alternating optimization (AO) algorithm  is proposed to decompose it into the  model selection and the joint rate and beamforming optimization subproblems, which are solved by the exhaustive search method and the combination of successive convex approximation (SCA) with fractional programming (FP), respectively.    
	\end{itemize} 
	
	The remainder of this paper is organized as follows. Section \uppercase\expandafter{\romannumeral2} introduces the system model. In Section \uppercase\expandafter{\romannumeral3}, we formulate the ISSC problem. In Section \uppercase\expandafter{\romannumeral4}, we  propose the algorithm to address  problem. Section  \uppercase\expandafter{\romannumeral5} provides the numerical simulation results. Section \uppercase\expandafter{\romannumeral6} concludes this paper.
	
	\emph{Notations}: 
	Bold upper-case and lower-case letters, e.g., $\mathbf{X}$ and $\mathbf{x}$, denote matrices and vectors, respectively.
	$\operatorname{Tr}(\cdot)$ and $\mathrm{vec}(\cdot)$ denote the trace and vectorization operators.
	$\mathbb{E}\left(\cdot\right)$ and $\mathbb{I}_{\{\cdot\}}$ indicate expectation and the indicator function.
	$\log_2(\cdot)$, $\log_{10}(\cdot)$, $|\cdot|$, $\|\cdot\|_F$, $\mathbf{I}_N$, $\min\{\cdot\}$, $j^2=-1$, and $\Re(\cdot)$ denote the base-2 logarithm, base-10 logarithm, Euclidean norm, Frobenius norm, identity matrix with $N$ dimensions, minimum operator, imaginary unit, and real-part operator, respectively.
	\begin{figure*}[!t] 
		\vspace{-1mm}
		\setlength{\belowcaptionskip}{-20pt}
		\centering
		\includegraphics[width=1\textwidth]{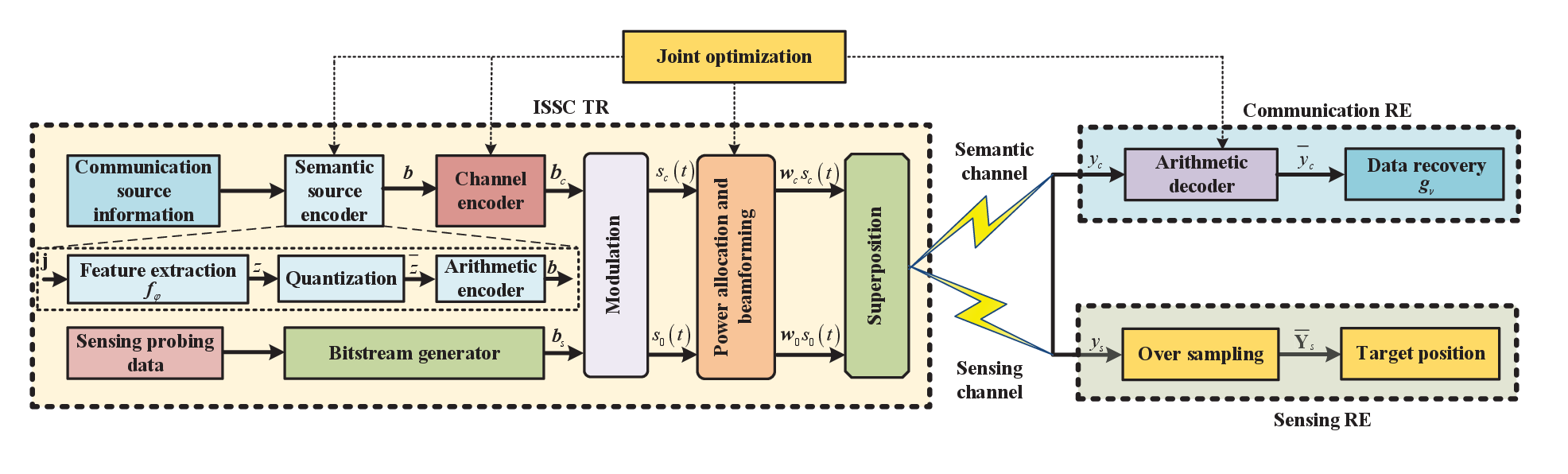} 
		\vspace{-10mm}   
		\caption{Framework of the ISSC System.}
		\label{fig:e}
	\end{figure*}
	\section{System Model}
	In this paper, we consider a bi-static ISSC system consisting of one transmitter (TR), one target, and two  REs, as shown in Fig. \ref{fig:e}. The TR  transmits the ISSC signal to both the target and  REs. The SRE receives the  signal reflected from the target for sensing task (e.g., position), and the CRE directly  receives the  signal from the TR for the SemCom task, (e.g., image reconstruction). Here,  the TR and the SRE are considered to be equipped with $N_T$ and $N_R$ antennas, respectively, while the CRE is considered to be equipped with a  single antenna. The positions of the  TR, target, CRE and  SRE are denoted as $\mathbf{p}_{b}=\left[b_x,b_y\right]^{T}$, $\mathbf{p}_{o}=\left[o_x,o_y\right]^{T}$, $\mathbf{p}_{c}=\left[c_x,c_y\right]^{T}$	and $\mathbf{p}_{r}=\left[r_x,r_y\right]^{T}$, respectively. 
	
	Based on the above setup, the communication source undergoes semantic source encoding, channel encoding,  modulation, and then is mapped into SemCom symbols. The sensing probing data is transformed into the sensing bitstream via a bitstream generator, and the bitstream is then transformed into the corresponding sensing symbols after modulation, as in the SemCom case. Then, both the SemCom and sensing symbols are superimposed into an ISSC signal via the  beamforming technique to perform tasks at the REs.   As  SemCom and sensing share a common TR  in the considered bi-static  ISSC system, it is necessary  to characterize the performance tradeoff by designing the coding rate and beamforming schemes for both the SemCom and sensing tasks. The detailed transmission and reception processes  for different tasks will be introduced in the sequel.
	\subsection{Transmission at  ISSC TR}
	\subsubsection{SemCom Task}
	For the SemCom task, the observation vector $\mathbf{j}\in\mathbb{R}^{M_1}$ is compressed at the TR   using the semantic source encoder, and   is finally reconstructed at the CRE. Specifically, the input data $\mathbf{j}$ is mapped to a $K$-dimensional vector $\boldsymbol{z}\in\mathbb{C}^{K}$  through the DNN-based feature extraction function $\boldsymbol{f}_{\varphi}$ with parameter $\varphi$, i.e., $\boldsymbol{z}=\boldsymbol{f}_{\varphi}(\mathbf{j})$. For lossy compression, $K$ is much smaller than $M_1$. Then, $\boldsymbol{z}$ is quantized as $\overline{\boldsymbol{z}}=\lfloor \boldsymbol{z}\rceil$, where $\lfloor\cdot \rceil$ denotes the uniform scalar quantization operation \cite{9242247}. Next, $\overline{\boldsymbol{z}}$ is encoded as a semantic  bitstream $\mathbf{b}$ containing semantic information  with a total length of $B$ based on its probability mass function $\boldsymbol{P}_{\overline{\boldsymbol{z}}}\left(\overline{\boldsymbol{z}}\right)$ by applying a lossless entropy coding method, i.e., arithmetic coding \cite{10.1145/214762.214771}. The expected source coding rate $R_s$ is calculated as $R_s\triangleq\mathbb{E}\left\{-\log_{2}\boldsymbol{P}_{\overline{\boldsymbol{z}}}\left(\overline{\boldsymbol{z}}\right)\right\}$. To protect against channel errors, the bitstream $\mathbf{b}$ is partitioned into $I=\lceil\frac{B}{N_c} \rceil$ blocks  and then  transformed into $\mathbf{b}_c$ via the block-wise channel coding, with each $N_c$-bit block being modulated (e.g., QPSK, QAM) into an $L$-length symbol vector. Subsequently, the SemCom symbol vector~$\mathbf{s}_c \in \mathbb{C}^N$ is obtained, with $N=IL$. 
	Hence, the channel coding rate is calculated as $R_{c}=\frac{N_c}{L}$, and the average number of channel uses per source sample is defined as $E_a=\frac{R_s}{R_c}$. The metric measuring the average channel uses per dimension of $\mathbf{j}$ is defined as the average bandwidth ratio (ABR), with $\text{ABR}=\frac{E_a}{M_1}$. 

	\subsubsection{Sensing Task}
	For the sensing task,  a bitstream generator is adopted to generate the sensing bitstream $\mathbf{b}_0$, which is also mapped into a symbol vector $\mathbf{s}_0\in\mathbb{C}^{N}$ through the same modulation process  used in the SemCom task.
	Then, both the SemCom and sensing symbols, i.e., $\mathbf{s}_c$ and $\mathbf{s}_0$,  
	are superimposed to form an ISSC signal via  beamforming technique to perform tasks at the REs. For the sensing task, the TR transmits the ISSC signal to the target, and then the SRE receives the echo reflected by the target to estimate the target position.   
	\subsubsection{Transmission Signal}
	Based on the above analysis, the ISSC signal transmitted by the TR in the $t$-th time slot is expressed as
	\begin{equation}\label{1}
		\mathbf{x}{\left(t\right)}=\sum_{u\in\mathcal{U}}\mathbf{w}_{u}(t){s}_{u}{\left(t\right)}, t\in\mathcal{T},
	\end{equation}
	where $\mathcal{U}=\{c,0\}$ is the set of subscripts, with $c$ and $0$ corresponding to the SemCom and sensing tasks, respectively; $\mathbf{w}_{u}(t)\in\mathbb{C}^{N_{T}\times 1}$ denotes the beamforming vector for the $u$-th task; $\mathcal{T}=\{1,\cdots,T\}$, with $T$ being the total number of  time slots; ${s}_u(t)$ is the extracted data from ${\bf s}_u$, i.e.,  
	\begin{equation}\label{2}
		{s}_{u}\left(t\right)=\sum_{n\in\mathcal{N}}{s}_{u}\left[n\right]g\left(t-\left(n-1\right)\Delta t\right),
	\end{equation}
	where ${s}_{u}\left[n\right]$ is the $n$-th symbol of ${\bf s}_u$, and $g(t)$ denotes the transmit pulse function with domain $g(t)\in[0,\Delta t]$, satisfying $\frac{1}{\Delta t}{\int_0^{\Delta t}|{g}(t)|^2dt}=1$\cite{11240213}. Notably, a rotationally invariant constellation $\mathcal{S}$ with zero mean and unit power is considered in our ISSC system, i.e.,
	\begin{equation}\label{3}
		\begin{aligned}
			&\mathbb{E}\left(\mathbf{s}_u\right)=0,  \mathbb{E}\left(\left|\mathbf{s}_u\right|^2\right)=1, \forall {\bf s}_u\in\mathcal{S}, u\in\mathcal{U},\\
		\end{aligned} 
	\end{equation}
	where  most commonly employed constellations, e.g., QPSK and QAM, satisfy these criteria in (\ref{3}) \cite{liu2025sensingcommunicationsignalsinformation}. Then, the transmission power at the ISSC TR is calculated as
	\begin{equation}\label{4}
		\mathbb{E}\left(\left|\left| \mathbf{x}{\left(
			t\right)}\right|\right|^{2}\right)=\operatorname{Tr}\left(
		\mathbf{W}_{c}+\mathbf{W}_0\right)
		=\operatorname{Tr}\left(\boldsymbol{\Sigma}\right)\leq\ P_{T},
	\end{equation}
	where $\boldsymbol{\Sigma}=\mathbf{W}_{c}+\mathbf{W}_{0}$, with $\mathbf{W}_{c}=\mathbf{w}_{c}\mathbf{w}_{c}^{H}\in\mathbb{C}^{N_{T}\times N_{T}}$ and $\mathbf{W}_{0}=\mathbf{w}_{0}\mathbf{w}_{0}^{H}\in\mathbb{C}^{N_{T}\times N_{T}}$ representing the covariance matrices for the SemCom and sensing tasks, respectively. Moreover, $P_T$ is the power budget at the TR.
	\vspace{-0.5cm}
	\subsection{Receptions at  Two REs}
	Based on the above analysis, the reception processes for the SemCom and sensing tasks are introduced in this section. 
	\subsubsection{Received Signal at  CRE}
	For the SemCom task, the TR directly transmits the signal ${\bf x}(t)$ given in (\ref{1}) to the CRE, and  the received signal $ y_{c}\left(t\right)$ at the CRE is given by
	\begin{align}\label{equ: received signal1}
		y_{c}\left(t\right)&=\mathbf{h}_{c}^{H}\mathbf{x}\left(t\right)+n\left(t\right)\nonumber\\
		&=\mathbf{h}_{c}^{H}\mathbf{w}_{c}{s}_{c}\left(t\right)+\mathbf{h}_{c}^{H} \mathbf{w}_{0}{s}_{0}\left(
		t \right)+n\left(t\right),
	\end{align}
	where $\mathbf{h}_{c}=\sqrt{\eta}_c\boldsymbol{\alpha}\left(\varphi\right)\in\mathbb{C}^{N_{T}\times 1}$ represents the channel vector of the direct link from the TR to the CRE, with $\eta_c=d_c^{-\epsilon_1}$ representing the path loss,  $\epsilon_1$ denoting the path loss exponent, and $d_c$ denoting the distance between the TR and the CRE; $\boldsymbol{\alpha}\left(\cdot\right)=\left[
	1, e^{j2\pi\frac{d_a}{\lambda}sin\left(\cdot\right)},\cdots,e^{j2\pi\left(N_{T}-1\right)\frac{d_a}{\lambda}sin\left(\cdot\right)}\right]^{T}$ is the steering vector, with $\lambda$ being the carrier wavelength, and $d_a$ being the distance between  two adjacent antennas;  $\varphi=\arctan\left(\frac{c_y-b_y}{c_x-b_x}\right)+\mathbb{I}_{\{c_x<b_x\}}\pi$ is the direction of departure (DoD), with $\mathbb{I}_{\{\cdot\}}$ being the indicator function; $n\left(t\right)$ is the   circularly symmetric complex Gaussian (CSCG) noise satisfying $n\left(t\right)\sim\mathcal{CN}\left(0,\sigma_{n}^{2}\right)$.
	Here, a block fading scenario is considered, where the channel coefficients remain constant within one transmission block and vary across different blocks.
	Based on (\ref{equ: received signal1}), the signal-to-interference-plus-noise ratio (SINR) at the CRE is given by
	\begin{equation}\label{6}
		\gamma=\frac{\left|\mathbf{h}_c^{H}\mathbf{w}_c\right|^2}{\left|\mathbf{h}_c^{H}\mathbf{w}_0\right|^2+\sigma_{n}^2}.
	\end{equation}
	
	Then, the decoded signal $\overline{y}_{c}$ is input to the DNN-based function $\boldsymbol{g}_{\nu}\left(\overline{y}_{c}\right)$  to recover the transmitted data $\mathbf{j}$, i.e., $\overline{\mathbf{j}}=\boldsymbol{g}_{\nu}\left(\overline{y}_{c}\right)$, where $\overline{\mathbf{j}}$ represents the recovered data. 
	From the well-known finite blocklength transmission theory, the lower bound of the average BER with random coding is approximated as \cite{li2025adaptivesourcechannelcodingsemantic}
	\begin{equation}\label{7}
		\rho_b\ge\frac{Q\left(\mathbf{ln}_2\frac{\sqrt{L}\left(C\left(\gamma\right)-R_c\right)}{B\left(\gamma\right)}\right)}{R_c L},
	\end{equation}
	where $Q\left(x\right)=\sqrt{\frac{1}{2\pi}}\int_{x}^{\infty}e^{-\frac{t^2}{2}}dt$, $B\left(\gamma\right)=\sqrt{1-\frac{1}{\left(1+\gamma\right)^2}}$, and $C\left(\gamma\right)=\text{log}_{2}\left(1+\gamma\right)$ denotes the channel capacity with $\gamma$ being given in (\ref{6}). It is obvious that $\rho_b$ is determined by $R_c$ and $\gamma$.

	\subsubsection{Received Signal at   SRE}
	For the sensing task, the TR transmits the ISSC signal ${\bf x}(t)$ to the target, and it is then reflected to the SRE.  Thus, the received signal at the  SRE is given by 
	\begin{align}\label{9}
		\mathbf{y}_{s}\left(t\right)=\mathbf{H}_{0}\mathbf{x}\left(t-\tau\right)+\mathbf{n}_s\left(t\right),	
	\end{align} 
	where $\mathbf{H}_{0}=\sqrt{\beta_{0}}\varrho\,\boldsymbol{\alpha}(\phi)\boldsymbol{\alpha}(\theta)^{H}\in\mathbb{C}^{N_{R}\times N_{T}}$ is the channel  matrix of the TR–target–SRE link, with $\beta_{0}=d_{s}^{-\epsilon_2}$ denoting the path loss, $d_{s}$ being the distance along the TR–target–SRE path, and $\epsilon_2$ being the path-loss exponent; $\varrho$ is the reflection coefficient; $\phi$ and $\theta$ denote the direction of arrival (DoA) at the SRE from the target and  the DoD from the TR to the target, respectively, and are defined as  $\phi=\arctan\left(\frac{o_y-r_y}{o_x-r_x}\right)+\mathbb{I}_{\{o_x<r_x\}}\pi$ and $\theta=\arctan\left(\frac{o_y-b_y}{o_x-b_x}\right)+\mathbb{I}_{\{o_x<b_x\}}\pi$;
	$\tau$ is the transmission delay of the TR--target--SRE link; $\mathbf{n}_{s}(t)$ is the received CSCG noise satisfying $\mathbf{n}_{s}(t)\sim\mathcal{CN}(0,\sigma_{s}^{2}\mathbf{I}_{N_{R}})$.
	As the TR and SRE are deployed at different positions, imperfect TS inevitably exists and cannot be ignored in the considered bi-static ISSC system\cite{10684491}. Consequently, the  transmission delay of the TR–target–SRE link under imperfect TS is modeled as 
	\begin{align}\label{10}
		\tau=\tau_{t,r}+\Delta\tau_{t,r},
	\end{align}
	where $\tau_{t,r}$ is the transmission delay under perfect TS, and  
	$\Delta	\tau_{t,r}$ is the additional TS error. Without loss of generality, $\Delta{\tau_{t,r}}$ is modeled as a Gaussian random variable satisfying $\Delta{\tau_{t,r}}\sim\mathcal{N}\left(0,\sigma_{T}^{2}\right)$, where the variance~$\sigma_T^2$  is known at the SRE\cite{10684491}.
	\vspace{-0.3cm} 
	\section{ISSC Problem Formulation}
	In this section,  an E2E distortion  is derived as the  performance metric for the SemCom task via a regression-based fitting method, and then an HCRB is obtained as the positioning performance metric for the sensing task under imperfect TS. To characterize the achievable region between SemCom and sensing, an E2E distortion minimization problem is formulated by considering the HCRB threshold, channel uses, and power budget.  
	\subsection{E2E Distortion  for  SemCom Task}
	Without loss of generality,  the distortion function $D_{o}$ is modeled as the sum of source and channel distortions\cite{10387242}, i.e., 
	\begin{equation}\label{11}
		D_{o}\approx D_{o}^{s}\left(R_s\right)+D_{o}^{c}\left(R_s, \rho_{b} \right), 
	\end{equation}
	where $D_{o}^{s}\left(R_s\right)$ and $D_{o}^{c}\left(R_s, \rho_{b} \right)$ represent the source  and channel distortions, respectively. Note that $D_{o}^{s}\left(R_s\right)$ varies with the source coding rate $R_s$, which is determined by the DNN parameters $\varphi$ of the semantic source encoder. $D_{o}^{c}\left(R_s, \rho_{b} \right)$ is jointly determined by the source coding rate $R_s$ and the BER $\rho_{b}$ defined in (\ref{7}). As the direct analytical modeling of channel distortion $D_{o}^{c}\left(R_s, \rho_{b} \right)$ by the closed-form expression is usually intractable due to its intricate dependence on both the unattainable Lipschitz constant  of the DNN  and  BER $\rho_{b}$ \cite{10845799}, a regression-based fitting method is utilized to model the E2E distortion $D_o$, where  logistic regression  is used to characterize its dependence  on the BER $\rho_{b}$. The detailed steps are introduced as follows. 
	\subsubsection{Model Design}
	$G$ hyper-prior-based DNN models, corresponding to a set of discrete source rates $\mathcal{R}_{s}=\left\{R_{s,1},...,R_{s,G}\right\}$,  are obtained by minimizing the objective function $\lambda D_{o} + R_{s}$ across varying $\lambda$ values under error-free conditions\footnote{All the DNN models undergo offline pre-training, obviating the requirement for online training and substantially alleviating computational demands.} \cite{Balle2018Variational}. Then, these hyper-prior-based DNN models are employed to implement feature extraction and image recovery. 
	\subsubsection{Distortion Measures}
	In the regression process, the E2E distortion  is defined as
	\begin{equation}
		D_o=\mathbb{E}_{\mathbf{j}}\left\{d_{o}\left(\mathbf{j},\hat{\mathbf{j}}\right)\right\},
	\end{equation} 
	where the  mean square error (MSE) is utilized to measure the distortion for the SemCom task, (e.g., image recovery), i.e., $	d_{o}\left(\mathbf{j},\overline{\mathbf{j}}\right)=\frac{1}{M_1}\left|\left|\mathbf{j}-\overline{\mathbf{
			j}}\right|\right|_2^2$. 
	\subsubsection{Distortion Approximation}
	We approximate the E2E distortion $D_o$ via the  regression-based fitting method on the Caltech-UCSD Birds-200-2011 (CUB-200-2011) dataset.
	Specifically, each image is compressed into a bitstream $\boldsymbol{b}$ by the semantic encoder, which is then encoded into $\boldsymbol{b}_c$ by the channel encoder. To simulate channel errors, the received bitstream $\hat{\boldsymbol{b}}_c$ is obtained by randomly flipping  bits in $\boldsymbol{b}_c$ according to a specified BER $\rho_b$.
	Then, the received bitstream $\hat{\boldsymbol{b}}_c$ is decoded to evaluate the reconstruction distortion.
	Finally, $D_o$ is obtained for each combination of the BER $\rho_b$ and source coding rate $R_s$
	by averaging the distortions over the entire dataset.

	Based on the above setup, we plot  the E2E distortion in logarithmic form, i.e., $\log_{10} D_{o,g}$ for $g\in\left\{1,2,...,G\right\}$, versus $\log_{10} \rho_{b}$ and $R_{s,g}$ in Fig.~\ref{2}, which shows  $\log_{10} D_{o,g}$ exhibits a sigmoidal behavior with respect to $\log_{10} \rho_{b}$\footnote{For simplicity, the subscript $g$ is omitted in the following analysis.}.  Consequently, a generalized logistic function is utilized to approximate $\log_{10} D_{o}$ as a function of $\log_{10} \rho_{b}$, i.e.,
	\begin{equation}\label{13_1}
		\log_{10} D_{o} \triangleq \hat{d_{o}^{s}}\left(R_{s}\right)+\frac{\hat{d_{o}^{c}}\left(R_s\right)}{1+e^{-a_{1}^{mse}\left(R_s\right)\left(\log_{10} \rho_{b}-a_{2}^{mse}\left(R_s\right)\right)}},
	\end{equation} 
	where $\hat{d_{o}^{s}}\left(R_s\right),\hat{d_{o}^{c}}\left(R_s\right), -a_{1}^{mse}\left(R_s\right)$, and $a_{2}^{mse}\left(R_s\right)$ are the fitting parameters   to be determined. 
	\begin{remark}
		As illustrated in Fig.~\ref{2},  when the BER $\rho_b$ falls below a certain threshold, the E2E distortion converges to a minimum floor and remains nearly constant. This is due to the fact that the transmission can  be regarded as error-free and the E2E distortion is primarily determined by the source distortion. Meanwhile, the distortion floor decreases with the increasing of source coding rate $R_s$. Hence, it is  concluded that  the source distortion  $D_{o}^{s}\left(R_s\right)$ decreases with the increasing of  the source coding rate $R_s$. Moreover,  it is also  observed  that
		a larger source coding rate $R_s$ steepens the distortion-BER slope, which demonstrates that a larger $R_s$ makes the system more sensitive to channel errors and thus increases the channel distortion. Based on the above findings, there exists a trade-off between source coding rate $R_s$ and BER $\rho_b$. 
	\end{remark}
	\vspace{-0.1cm}
	\subsection{HCRB Derivation for Sensing Task}
	For the sensing task, the target position $\mathbf{p}_o=\left[o_x,o_y\right]^{T}$ is the parameter of interest, while the TS error $\Delta \tau_{t,r}$ is modeled as a random nuisance parameter with a known prior distribution. To incorporate the influence of the TS error into the positioning bound, we define the  hybrid parameter vector as $\boldsymbol{\chi}=\left[\mathbf{p}_o^{T}, \Delta \tau_{t,r}\right]$. Then, the HCRB, a theoretical limit on the estimation variance of $\mathbf{p}_o$, is adopted as the positioning performance metric in the considered bi-static ISSC system. 
		\begin{figure}[!t]
		\setlength{\belowcaptionskip}{-23pt}
		\setlength{\abovecaptionskip}{0.1cm}
		\centering
		\vspace{-1em}
		\includegraphics[width=2.7in]{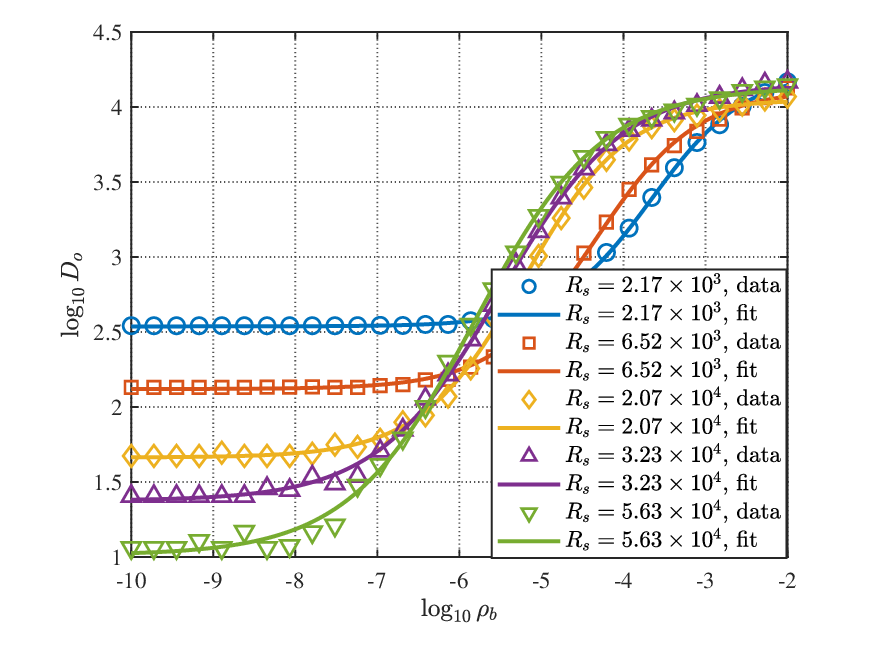}
		\caption{E2E distortion $D_o$ with respect to $\rho_b$ under different $R_s$.}
		\label{2}
	\end{figure}
	To simplify the analysis, the continuous received signal $\mathbf{y}_{s}\left(t\right)$ given in (\ref{9}) is 
	sampled at an interval of $\triangle t$, and then $N$ independent observations $\overline{\mathbf{Y}}_s=\left[\overline{\mathbf y}_s\left(t_1\right),\overline{\mathbf y}_s\left(t_2\right),...,\overline{\mathbf y}_s\left(t_N\right)\right]$ are obtained at the SRE, with $t_n=n\Delta t$ and $n\in\mathcal{N}=\left\{1,2,\cdots,N\right\}$. For analytical convenience, the sampled signal is further mapped into the frequency domain via the discrete Fourier transform (DFT), and then the signal in the frequency domain corresponding to the $k$-th $\left(k\in\left[1,2,\dots,N\right]\right)$ sample is given by
	\begin{equation}\label{15}
		\widetilde{\mathbf{y}}_{s}\left(w_k\right)=\mathbf{H}_{0}\widetilde{\mathbf{x}}\left(w_k\right)+\widetilde{\mathbf{n}}_{s}\left(
		w_k\right),
	\end{equation} 
	where $\widetilde{\mathbf{y}}_{s}\left(w_k\right)$ is the $k$-th frequency-domain sample of the received signal, and $\widetilde{\mathbf{x}}\left(w_k\right)=e^{-jw_{k}\tau}\left(
	\mathbf{w}_{c}\widetilde{s}_c\left(w_k\right)+\mathbf{w}_{0}\widetilde{s}_0\left(w_k\right)\right)$, with $w_{k}=\frac{2\pi \left(k-1\right)}{N\Delta t}$. Here, $\widetilde{s}_c\left(w_k\right)=\frac{1}{\sqrt{N}}\sum_{n=1}^{N}s_{c}\left[n\right]e^{-j2\pi \left(k-1\right)\frac{n-1}{N}}$, $\widetilde{s}_0\left(w_k\right)=\frac{1}{\sqrt{N}}\sum_{n=1}^{N}s_{0}\left[n\right]e^{-j2\pi \left(k-1\right)\frac{n-1}{N}}$, and $\widetilde{\mathbf{n}}_{s}\left(
	w_k\right)$  denote the DFT of $s_c\left(t_k\right)$, $s_0\left(t_k\right)$ and $\mathbf{n}_{s}\left(t_k\right)$, respectively. Then, the  received signal at the SRE in the frequency domain is rewritten as $\widetilde{\mathbf{Y}}_s=\mathbf{H}_{0}\widetilde{\mathbf{X}}+\widetilde{\mathbf{N}}_{s}$, where $\widetilde{\mathbf{X}}=\left[\widetilde{\mathbf{x}}\left(w_1\right),\widetilde{\mathbf{x}}\left(w_2\right),...,\widetilde{\mathbf{x}}\left(w_N\right)\right]\in\mathbb{C}^{N_{T}\times N}$and $\widetilde{\mathbf{N}}_s=\left[\widetilde{\mathbf{n}}_s\left(w_1\right),\widetilde{\mathbf{n}}_s\left(w_2\right),...,\widetilde{\mathbf{n}}_s\left(w_N\right)\right]\in\mathbb{C}^{N_{R}\times N}$. For simplicity, the frequency-domain signal $\widetilde{\mathbf{Y}}_s$ is vectorized as $\widehat{\mathbf{y}}_{s}=	\widehat{\mathbf{x}}+\widehat{\mathbf{n}}_{s}$, where $\widehat{\mathbf{y}}_{s}=\text{vec}\left(\widetilde{\mathbf{Y}}_s\right)$, $\widehat{\mathbf{x}}=\text{vec}\left(\mathbf{H}_{0}\widetilde{\mathbf{X}}\right)$, and $\widehat{\mathbf{n}}_{s}=\text{vec}\big(\widetilde{\mathbf{N}}_s\big)\sim\mathcal{CN}\left(\mathbf{0},\boldsymbol{\Xi}\right)$, with $\boldsymbol{\Xi}=\sigma_{s}^{2}\mathbf{I}_{N_{R}N}$. 
	
	As the TR and SRE are deployed at different places, the TS error $\Delta
	\tau_{t,r}$ is incorporated into the conventional CRB to characterize the positioning performance in the considered ISSC system by introducing the prior FIM of $\Delta
	\tau_{t,r}$ \cite{6877741}. Hence, we define the hybrid Fisher information matrix (FIM)  as the sum of the observed FIM $\mathbf{J}_{D}\left(\boldsymbol{\chi}\right)$ and the prior FIM $\mathbf{J}_{B}\left(\boldsymbol{\chi}\right)$, i.e., $\mathbf{J}\left(\boldsymbol{\chi}\right)=\mathbf{J}_{D}\left(\boldsymbol{\chi}\right)+\mathbf{J}_{B}\left(\boldsymbol{\chi}\right)$. Here, the expression of  $\mathbf{J}_{D}\left(\boldsymbol{\chi}\right)$ is given by \cite{10684491}
	\begin{align}\label{18}
		\mathbf{J}_D\left(\boldsymbol{\chi}\right)=-\mathbb{E}_{\widehat{\mathbf{y}}_{s}|\boldsymbol{\chi}}\left(\frac{\partial^{2}\ln\hat{\boldsymbol{f}}\left(\widehat{\mathbf{y}}_{s}|\boldsymbol{\chi}\right)}{\partial\boldsymbol{\chi}\partial\boldsymbol{\chi}^T}\right),
	\end{align}
	where $\hat{\boldsymbol{f}}\left(\widehat{\mathbf{y}}_{s}|\boldsymbol{\chi}\right)$ denotes the conditional probability density function of the observation $\widehat{\mathbf{y}}_{s}$, i.e.,
	\begin{equation}\label{19}
		\hat{\boldsymbol{f}}\left(\widehat{\mathbf{y}}_{s}|\boldsymbol{\chi}\right)=\frac{\text{exp}\left\{-\overline{\mathbf{y}}_s^{H}\boldsymbol{\Xi}^{-1}\overline{\mathbf{y}}_s\right\}}{\pi^{N_{R}N}\det\left(\boldsymbol{\Xi}\right)}.
	\end{equation}
	where $\overline{\mathbf{y}}_s=\widehat{\mathbf{y}}_{s}-\widehat{\mathbf{x}}$. In addition, the prior FIM $\mathbf{J}_{B}\left(
	\boldsymbol{\chi}\right)$ is expressed as $\mathbf{J}_{B}\left(\boldsymbol{\chi}\right)=-\mathbb{E}_{\Delta\tau_{t,r}}\left(\frac{\partial^{2}\ln\rho\left(\Delta \tau_{t,r}\right)}{\partial\boldsymbol{\chi}\partial\boldsymbol{\chi}^T}\right)$,
	where $\rho(\Delta	\tau_{t,r})$ denotes the prior distribution of $\Delta \tau_{t,r}$, given by $\rho(\Delta	\tau_{t,r}) = \exp(-\Delta	\tau_{t,r}^2 / (2\sigma_T^2)) / (\sqrt{2\pi}\sigma_T)$.
	Based on the above analysis, we present the following proposition for deriving  the hybrid FIM $\mathbf{J}\left(\boldsymbol{\chi}\right)$.
		\begin{figure*}[b] 
		\hrulefill  
		\begin{align}
			\text{HCRB}(\mathbf{p}_o) &=\left(\mathbf{J}_{D}^{c_{1}}-\mathbf{J}_{D}^{c_{2}}\left(\mathbf{J}^{c_{3}}\right)^{-1}\left(\mathbf{J}_{D}^{c_{2}}\right)^{T}\right)^{-1}\tag{18}\label{241}\\
			&=\underbrace{\Bigl(\mathbf{J}_{D}^{c_{1}}\Bigr)^{-1}}_{\text{CRB}\left(\mathbf{p}_o\right)}
			+\underbrace{\Bigl(\mathbf{J}_{D}^{c_{1} }\Bigr)^{-1} 
				\left(\mathbf{J}_{D}^{c_{2} }\right)
				\biggl[ \mathbf{J}^{c_3} 
				- \Bigl(\mathbf{J}_{D}^{c_{2} }\Bigr)^{{T}}
				\Bigl(\mathbf{J}_{D}^{c_{1} }\Bigr)^{-1}
				\mathbf{J}_{D}^{c_{2}} \biggr]^{-1}  
				\Bigl(\mathbf{J}_{D}^{c_{2} }\Bigr)^{{T}}
				\Bigl(\mathbf{J}_{D}^{c_{1} }\Bigr)^{-1}}_{\mathrm{additional\ error\ caused \ by\ TS\ error}}\tag{19}\label{242}.
		\end{align}
		\vspace{-3em}
	\end{figure*}
	\begin{Proposition}\label{3.1}
		The hybrid FIM $\mathbf{J}\left(\boldsymbol{\chi}\right)$ can be equivalently expressed in the following block matrix form by taking the first-order partial derivatives with respect to $\boldsymbol{\chi}$, i.e., 
		\begin{equation}\label{22}
			\mathbf{J}(\boldsymbol{\chi}) =
			\underbrace{
				\begin{bmatrix}
					\mathbf{J}_D^{c_1} & \mathbf{J}_D^{c_2} \\
					\left(\mathbf{J}_D^{c_2}\right)^{T} & J_D^{c_3}
				\end{bmatrix}
			}_{\mathbf{J}_D(\boldsymbol{\chi})}
			+ \mathbf{J}_B(\boldsymbol{\chi}),
		\end{equation}
		where $\mathbf{J}_B\left(\boldsymbol{\chi}\right)=\begin{bmatrix}
			\mathbf{0}_{2\times2}&\mathbf{0}_{2\times1} \\
			\mathbf{0}_{1\times2}&\frac{1}{{\sigma_{T}}^2}
		\end{bmatrix}$. The submatrices in $\mathbf{J}_D\left(\boldsymbol{\chi}\right)$ are given by
		\begin{align}\label{23}
			\mathbf{J}_D^{c_1}=
			\begin{bmatrix}
				J_D^1&J_D^2 \\
				J_D^2&J_D^3
			\end{bmatrix},		                   \mathbf{J}_D^{c_2}=\begin{bmatrix}
				J_D^4\\J_D^5
			\end{bmatrix},
			\mathbf{J}_D^{c_3}=J_D^6,
		\end{align}
		where  $\mathbf{J}_D^{c_1}=\mathbf{J}_D^{\mathbf{p}_o\mathbf{p}_o}$, $\mathbf{J}_D^{c_2}=\mathbf{J}_D^{\mathbf{p}_o\Delta \tau_{t,r}}$, $\mathbf{J}_D^{c_3}=\mathbf{J}_D^{\Delta \tau_{t,r}\Delta \tau_{t,r}}$, and  $J_D^j$ is derived as $J_D^j=\frac{2}{\sigma_{s}^2}\mathcal{R}\left\{\operatorname{Tr}\left(\boldsymbol{\zeta}_{j}\boldsymbol{\Sigma}\right)\right\}$, with  $j\in\boldsymbol{\mathcal{J}}=\left\{1,...,6\right\}$. Here,  each $\boldsymbol{\zeta}_j$ is derived as  $\boldsymbol{\zeta}_{1} =\sum_{n=1}^{N} \left( \dot{\mathbf{H}}_{0}^{n} \right)^{H}  \dot{\mathbf{H}}_{0}^{n}$, $\boldsymbol{\zeta}_{2} =\sum_{n=1}^{N}\left( \dot{\mathbf{H}}_{0}^{n} \right)^{H}  \ddot{\mathbf{H}}_{0}^{n}$, $\boldsymbol{\zeta}_{3} =\sum_{n=1}^{N} \left( \ddot{\mathbf{H}}_{0}^{n} \right)^{H}  \ddot{\mathbf{H}}_{0}^{n}$, $\boldsymbol{\zeta}_{4}=\sum_{n=1}^{N} -j w_{n}\left( \dot{\mathbf{H}}_{0}^{n} \right)^{H} \mathbf{H}_{0}$, $	\boldsymbol{\zeta}_{5}=\sum_{n=1}^{N} -j w_{n}\left( \ddot{\mathbf{H}}_{0}^{n} \right)^{H} \mathbf{H}_{0}$, and $\boldsymbol{\zeta}_{6} =\sum_{n=1}^{N} w_{n}^{2}\left( \mathbf{H}_{0} \right)^{H} \mathbf{H}_{0}$, where
		$\dot{\mathbf{H}}_{0}^{n}=\mathbf{H}_{0}^{'}-jw_{n}\tau^{'}\mathbf{H}_{0}$ and $\ddot{\mathbf{
				H}}_0^n=\mathbf{H}_{0}^{''}-jw_{n}\tau^{''}\mathbf{H}_0$. Moreover,  $\mathbf{H}_{0}^{'}$ and $\tau^{'}$ denote the first-order derivatives of $\mathbf{H}_0$ and $\tau$ with respect to $o_x$, respectively, while $\mathbf{H}_{0}^{''}$ and $\tau^{''}$ denote the first-order derivatives of $\mathbf{H}_0$ and $\tau$ with respect to   $o_y$.
		\vspace{0.1cm}
		\begin{proof}
			Please see Appendix A. \hfill $\blacksquare$
		\end{proof}
			\vspace{0.1cm}
	\end{Proposition}

		\begin{Proposition}
			Based on the above analysis in Proposition 3.1, the HCRB matrix  of the target position, is derived at the bottom of this page,
			where (\ref{241}) is obtained by applying the block matrix inversion formula to the hybrid FIM in (\ref{22}), and  (\ref{242}) is obtained by applying the Woodbury matrix identity.
			Here, $\mathbf{J}^{c_3}=\mathbf{J}_{D}^{c_3}+\frac{1}{\sigma_{T}^2}$, and the second term in (\ref{242}) is the additional  error for the  target position $\mathbf{p}_o$ caused by imperfect TS.  
		\end{Proposition}
		\subsection{Problem Formulation}
		We define the achievable  region $\mathcal{A}\left(\Phi, \Pi\right)$ to describe the achievable region between the E2E distortions $D_o$ and the HCRB for the SemCom and sensing task, i.e.,
		\setcounter{equation}{19}
		\begin{equation}\label{201}
			\resizebox{0.88\linewidth}{!}{$
				\begin{aligned}[b]
					\mathcal{A}(\Phi, \Pi) = 
					\bigg\{ 
					& (\Phi, \Pi): D_o \leq \Phi, \operatorname{Tr}\!\big( \text{HCRB}(\mathbf{p}_o) \big) \leq \Pi, \\
					& \operatorname{Tr}(\boldsymbol{\Sigma}) \le P_T, \frac{R_s}{R_c} \le E_a^{\max}, R_c \le C(\gamma)
					\bigg\},
				\end{aligned}
				$}
		\end{equation}
		where  $(\Phi, \Pi)$ is the distortion-HCRB pair;  $\Pi$ and $E_a^{\text{max}}$ denote the  HCRB threshold and the maximum channel uses, respectively; $C(\gamma)$ is the channel capacity given in (\ref{7}). Then, to characterize the achievable region in (\ref{201}),  an  E2E distortion  minimization problem is formulated, subject to   the HCRB threshold, channel uses, and power budget, i.e., 
		\begin{align}
			\left(\text{P3.1}\right)\mathop{\text{min}}\limits_{\left\{R_s\in\mathcal{R}_s, R_c,\mathbf{w}_0,\mathbf{w}_c\right\}} ~~~~&~D_o \label{261}\\
			\text{s.t.}~~~~~~~~~~&~\operatorname{Tr}\left\{\text{HCRB}\left(\mathbf{p}_o\right)\right\} \leq\Pi \label{2711}, \\
			~~~~&~\frac{R_s}{R_c}\leq E_a^{\text{max}}\label{27},\\
			~~~~&~  R_c\leq C\left(\gamma\right) \label{28},\\
			~~~~&~\operatorname{Tr}\left(\boldsymbol{\Sigma}
			\right)\leq\ P_{T} \label{29},
		\end{align}	
		where the objective function  (\ref{261}) and the constraint  (\ref{2711}) are non-convex due to the coupled design variables. Moreover, since the design variable $R_s$ is discrete, Problem (P3.1) is a mixed-integer program, which is also intractable in general. Therefore, it is necessary to design an efficient algorithm to address this problem.  
		\section{Algorithms}
		In this section, an AO algorithm  is proposed to decompose the original problem into the model selection subproblem, which is solved by exhaustive search method, and the  joint rate and beamforming optimization subproblem, which is solved by the combination of  FP and SCA method.  
		\subsection{Model Selection}
		For fixed channel rate $R_c$ and beamforming vectors $\mathbf{w}_0$ and $\mathbf{w}_c$, the original Problem (P3.1) is simplified into a  model selection subproblem. Since each candidate DNN model is uniquely characterized by a specific source coding rate $R_s$, the model selection process is equivalent to finding the optimal $R_s$ from the discrete set $\mathcal{R}_s$, i.e., 
		\begin{align}
			\left(\text{P4.1}\right)\mathop{\text{min}}\limits_{\left\{R_s\in\mathcal{R}_s\right\}} ~~~~&~D_o \label{26}\\
			\text{s.t.}
			~~~~~~&~\frac{R_s}{R_c}\leq E_a^{\text{max}}\label{271},
		\end{align}	
		where Problem (P4.1) is solved by the well-known exhaustive search method to obtain the optimal value of $R_s$   satisfying  constraint  (\ref{271}).
				\begin{figure*}[b] 
			\hrulefill    
			\begin{align}
				U_1\left(R_c\right)&=\log_{10}Q\left(\hat{	\psi}\right)-\log_{10}\left(R_cL\right)-\hat{a}\left(\hat{	\psi}\right)\left(\psi-\hat{	\psi}\right)\log_{10}^{e}\label{40},~~~~~~~~~~~~~~~~~~~~~~~~~~~~\tag{38}
			\end{align}
			\begin{equation}
				\begin{aligned}
					~~~~~~U_{2}\left(\nu,\hat{\rho}_b\right)&=-e^{\nu^{'}-	\beta_1^{'}}\left(\nu-\nu^{'}-a_1^{mse}\left(R_s\right)\left(\hat{\rho}_b-\hat{\rho}_b^{'}+1\right)\right)
					-e^{\nu^{'}}\left(\nu-\nu^{'}+1\right)+\hat{d_{o}^{c}}\left(R_s\right)\label{41}.
				\end{aligned}\tag{39}
			\end{equation}
			\vspace{-2em}
		\end{figure*}
		\subsection{Joint Rate and Beamforming Optimization}
		For fixed $R_s$, Problem (P3.1)  is transformed into a continuous optimization problem, i.e.,
		\begin{align}
			\left(\text{P4.2}\right)\mathop{\text{min}}\limits_{\left\{R_c,\mathbf{w}_0, \mathbf{w}_c\right\}} ~~~~&~D_o \label{d0}\\
			\text{s.t.} 
			~~~~~~~~&~ R_c\geq \frac{R_{s}}{E_a^{\text{max}}} \label{d1}, \\
			~~~~&~\text{(\ref{2711})},\text{(\ref{28})}, \text{(\ref{29})} \label{31}. 
		\end{align}	
		where constraint (\ref{d1}) is obtained by rewriting  constraint (\ref{27}) with the fact of $R_c>0$. It is observed that
		Problem~(P4.2) is still non-convex due to the complicated expression in (\ref{d0}) and the coupled design variables in (22).
		Then, Problem (P4.2) is further decomposed into a coding rate subproblem and a beamforming optimization subproblem.
		\subsubsection{Coding Rate Optimization}
		For simplicity, two logarithmic auxiliary variables $\hat{\rho}_b$ and $\nu$, satisfying  $\hat{\rho}_b=\text{log}_{10}\rho_{b}$ and $e^{\nu}=\frac{\hat{d_{o}^{c}}\left(R_s\right)}{1+e^{-a_{1}^{mse}\left(R_s\right)\left(\hat{\rho}_b-a_{2}^{mse}\left(R_s\right)\right)}}$, are introduced, respectively. Then, with fixed $\mathbf{w}_0$ and $\mathbf{w}_c$,  Problem (P4.2) is reformulated as
		\begin{align}
			\left(\text{P4.3}\right)\mathop{\text{min}}\limits_{\left\{R_c,\nu,\hat{\rho}_b\right\}} ~~~~&~10^{\hat{d_{o}^{s}}\left(R_s\right)+e^{\nu}}\label{34_1}\\
			\text{s.t.} 
			~~~~~~~&~U\left(R_c\right)\leq \hat{\rho}_b\label{35_1},\\
			~~~~&~\frac{\hat{d_{o}^{c}}\left(R_s\right)}{1+e^{-\beta_1}}\leq e^{\nu}\label{36_1},\\
			~~~~&~\text{(\ref{28})}, \text{(\ref{d1})}\label{37_1},
		\end{align}
		where (\ref{34_1}) is obtained by substituting $\hat{\rho}_b$ and $e^{\nu}$  into (\ref{d0}). Moreover, constraints (\ref{35_1}) and (\ref{36_1}) are obtained by introducing auxiliary variables $\hat{\rho}_b$ and $e^{\nu}$, where   $U\left(R_c\right)=\log_{10}\left(\frac{Q\left(\mathbf{ln}_2\frac{\sqrt{L}\left(C\left(\gamma\right)-R_c\right)}{B\left(\gamma\right)}\right)}{R_c L}\right)$ and $\beta_1=a_1^{mse}\left(R_s\right)\left(\hat{\rho}_b-a_2^{mse}\left(R_s\right)\right)$. To facilitate the analysis,  constraint (\ref{36_1}) is rewritten as
		\begin{equation}
			-e^{\nu-\beta_1}-e^{\nu}+\hat{d}_o^{c}\left(R_s\right)\leq 0\label{36_11},\\
		\end{equation}
		\noindent where (\ref{36_11}) holds due to the fact that $1+e^{-\beta_1}>1$.
		It is observed that (\ref{34_1}) monotonically increases with respect to $\hat{\rho}_b$ and $e^{\nu}$.  Consequently, all constraints in Problem (P4.3) must hold with equality for the optimal solution; otherwise, (\ref{34_1}) can be further reduced by decreasing $\hat{\rho}_b$ and $e^{\nu}$, which confirms the Problem (P4.2) is equivalent to Problem (P4.3) . To address the  non-convexity of constraints    (\ref{35_1}) and (\ref{36_11}), we derive the convex upper bounds for the left-hand side (LHS) expressions of  constraints (\ref{35_1}) and (\ref{36_11})  in the following lemma. 

		\begin{Lemma}
			For  the LHS of (\ref{35_1}) and (\ref{36_11}), the convex upper bounds  derived via  the upper bound of Q-function and the first-order Taylor expansion, respectively, i.e., 
			\begin{align}
				U\left(R_c\right)&\leq U_1\left(R_c\right)\label{38},\\
				-e^{\nu-\beta_1}-e^{\nu}+\hat{d}_o^{c}\left(R_s\right)&\leq  U_{2}\left(\nu,\hat{\rho}_b\right)\label{39},
			\end{align}
			where $U_1\left(R_c\right)$ and $U_2\left(\nu,\hat{\rho}_b\right)$ are given in (\ref{40}) and (\ref{41}), respectively,
			with  $\psi=\frac{\sqrt{L}\left(C\left(\gamma\right)-R_c\right)\text{ln}2}{B\left(\gamma\right)}$, $\hat{\psi}=\frac{\sqrt{L}\left(C\left(\gamma\right)-R_c^{'}\right)\text{ln}2}{B\left(\gamma\right)}$, $e^{\nu^{'}}={\hat{d_{o}^{c}}\left(R_s\right)}/{1+e^{-a_{1}^{mse}\left(R_s\right)\left(\hat{\rho}_b^{'}-a_{2}^{mse}\left(R_s\right)\right)}}$, $\beta_1^{'}=a_1^{mse}\left(R_s\right)\left(\hat{\rho}_b^{'}-a_2^{mse}\left(R_s\right)\right)$, and $\hat{\rho}_b^{'}=Q\left(\hat{\psi}\right)\big/ \left(R_c^{'} L\right)$. Here, $R_c^{'}$, $\hat{\rho}_b^{'}$, and  $\nu^{'}$ are fixed and   obtained in the previous iteration, respectively.
		\end{Lemma}
		\begin{proof}
			Please see Appendix B. \hfill $\blacksquare$
		\end{proof}
		
		Replacing the LHSs of constraints (\ref{35_1}) and (\ref{36_11}) with their convex upper bounds (\ref{40}) and (\ref{41}), respectively, Problem (P4.3) is rewritten as
		\setcounter{equation}{39}
		\begin{align}
			\left(\text{P4.4}\right)\mathop{\text{min}}\limits_{\left\{R_c,\nu,\hat{\rho}_b\right\}}~~~~&~10^{\hat{d_{o}^{s}}\left(R_s\right)+e^{\nu}}\label{42}\\
			\text{s.t.} 
			~~~~~~~&~U_1\left(R_c\right)- \hat{\rho}_b\leq 0\label{43},\\
			~~~~&~U_{2}\left(\nu,\hat{\rho}_b\right)\leq 0\label{44},\\
			~~~~&~\text{(\ref{28})}, \text{(\ref{d1})}\label{45},
		\end{align} 
		which is convex and can be solved by  some optimization tools, e.g., CVX \cite{cvx_software}.
		\subsubsection{ Beamforming Optimization}
		With fixed $R_c$, the objective function $D_o$ is monotonically increasing with  $\rho_b$ according to (\ref{13_1}). Consequently, Problem (P4.2) is equivalently transformed into a BER $\rho_{b}$ minimization problem, i.e.,
		\begin{align}
			\left(\text{P4.5}\right)\mathop{\text{min}}\limits_{\left\{\mathbf{w}_0, \mathbf{w}_c\right\}} ~~~~&~\frac{Q\left(\hat{\gamma}\right)}{R_c L}\label{46}\\
			\text{s.t.} 
			~~~~~~&~\text{(\ref{2711})}, \text{(\ref{29})}\label{47},
		\end{align}	
		where (\ref{46}) is obtained according to (\ref{7}), with $\hat{\gamma}=\ln_2\frac{\sqrt{L}\left(C\left(\gamma\right)-R_c\right)}{\sqrt{1-\frac{1}{\left(1+\gamma\right)^2}}}$. It is intractable to directly solve Problem (P4.5) due to the variable coupling within the Q-function. Noting that the Q-function is monotonic, the monotonicity of BER is  determined by $\hat{\gamma}$. Hence,  it is necessary to analyze the monotonicity of $\hat{	\gamma}$ with  respect to $\gamma$ to  simplify this problem. Then, we have the following lemma:
		\begin{Lemma}\label{lemma 4.2}
			The objective function in Problem (P4.5) monotonically decreases with the increasing of  $\gamma$. 
		\end{Lemma}
		\begin{proof}
			Please see Appendix C. \hfill $\blacksquare$  
		\end{proof}
		
		Based on Lemma \ref{lemma 4.2}, Problem (P4.5) is transformed into the SINR $\gamma$ maximization problem, i.e.,
		\begin{align}
			\left(\text{P4.6}\right)\mathop{\text{max}}\limits_{\left\{\mathbf{w}_0, \mathbf{w}_c\right\}} ~~~~&~\frac{\left|\mathbf{h}_c^{H}\mathbf{w}_c\right|^2}{\left|\mathbf{h}_c^{H}\mathbf{w}_0\right|^2+\sigma_{n}^2}\label{50}\\ 
			\text{s.t.} ~~~~~~&~\text{(\ref{2711})}, \text{(\ref{29})}\label{51},
		\end{align}	
		where (\ref{50}) is still non-convex due to the coupling  design variables in the fractional structure. Therefore,   the FP method is utilized to  maximize a surrogate lower-bound function of  (\ref{50})\cite{11114787}, i.e., 
		\begin{align}
			2\Re\left(r\mathbf{h}_c^{H}\mathbf{w}_c\right)-\left|r\right|^2\left(\left|\mathbf{h}_c^{H}\mathbf{w}_0\right|^2+\sigma_{n}^2\right)\label{52},
		\end{align}
		where $  r = \mathbf{w}_c^H \mathbf{h}_c / (|\mathbf{h}_c^H \mathbf{w}_0|^2 + \sigma_n^2)$. Since   constraint  (\ref{2711}) is still non-convex, we utilize the positive semi-definite   property  of the FIM matrix $\mathbf{J}_{D}^{c_1}-\mathbf{J}_{D}^{c_2}\left(\mathbf{J}^{c_3}\right)^{-1}\left(\mathbf{J}_{D}^{\mathbf{c}_2}\right)^T$ and the property that $\operatorname{Tr}\left(\mathbf{U}^{-1}\right)$ is a decreasing function on the positive semi-definite matrix space with a given matrix $\mathbf{U}$ to transform  the non-convex constraint  (\ref{2711}) into the following form by introducing an auxiliary matrix $\mathbf{\Omega}\in\mathbb{C}^{2\times2}$, i.e.,  
		\begin{equation}\label{54}
			\operatorname{Tr}\left(\mathbf{\Omega}^{-1}\right)\leq\Pi, \mathbf{\Omega}\succeq\mathbf{0},
		\end{equation}
		\begin{equation}\label{55}
			\mathbf{J}_{D}^{c_{1}}-\mathbf{J}_{D}^{c_{2}}\left(\mathbf{J}^{c_{3}}\right)^{-1}\left(\mathbf{J}_{D}^{c_{2}}\right)^{T} \succeq \mathbf{\Omega}.
		\end{equation}  
		By utilizing the Schur complement \cite{10050406}, (\ref{55}) is rewritten as
		\begin{figure*}[b] 
			\hrulefill   
			\begin{equation}\label{621}
				\begin{aligned}
					G_j\left(\mathbf{w}_c, \mathbf{w}_0\right)= &2\left|\left(\mathbf{w}_c^{'}\right)^{H}\left(\boldsymbol{\zeta}_j\right)^{H}\mathbf{w}_c \right|^2+2\left|\left(\mathbf{w}_0^{'}\right)^{H}\left(\boldsymbol{\zeta}_j\right)^{H}\mathbf{w}_0 \right|^2\\
					&+2\Re\left\{\left(\mathbf{a}_{1,j}^{'}-\left(\mathbf{w}_0^{'}\right)^H\mathbf{\Lambda}_j\right)\mathbf{w}_0\right\}+2\Re\left\{\left(\mathbf{a}_{2,j}^{'}-\left(\mathbf{w}_c^{'}\right)^H\mathbf{\Lambda}_j\right)\mathbf{w}_c\right\},
				\end{aligned}\tag{62}
			\end{equation}
			\vspace{-3em}
		\end{figure*}
		\begin{equation}\label{56}
			\begin{bmatrix}
				\mathbf{J}_D^{\mathbf{c}_1}-\mathbf{\Omega}&\mathbf{J}_D^{\mathbf{c}_2} \\
				\left(\mathbf{J}_D^{\mathbf{c}_2}\right)^T&\mathbf{J}^{c_3}
			\end{bmatrix}\succeq \mathbf{0}.
		\end{equation}
		
		It is clear that constraint (\ref{54}) is convex, while constraint (\ref{56})  is still non-convex. Hence, an auxiliary variable vector $\boldsymbol{\mu}=\left[\mu_1,\mu_2,\mu_3,\mu_4,\mu_5,\mu_6\right]^{T}\in\mathbb{C}^{6\times1}$ is introduced to extract $\mathbf{w}_0$ and $\mathbf{w}_c$ from (\ref{56}), and   Problem (P4.6) is reformulated as
		\begin{align}
			\left(\text{P4.7}\right)\mathop{\text{max}}\limits_{\left\{r,\boldsymbol{\Omega}, \mathbf{w}_0, \mathbf{w}_c,\boldsymbol{\mu}\right\}} ~&~ \text{(\ref{52})}~~~\label{57}\\
			\text{s.t.} ~~~~~~~&\begin{bmatrix}
				\frac{2}{\sigma_{s}^{2}}\Re\left(\mathbf{a}_1\right)-\mathbf{\Omega}&\frac{2}{\sigma_{s}^{2}}\Re\left(\mathbf{a}_2\right)\\
				\frac{2}{\sigma_{s}^{2}}\Re\left(\mathbf{a}_2^{T}\right)&\frac{2}{\sigma_{s}^{2}}\Re\left(\mathbf{a}_3\right)+\frac{1}{\sigma_{T}^2}
			\end{bmatrix}\succeq\mathbf{0}\label{58},\\
			~~~~&~\mu_{j}=f_j\left(\mathbf{w}_c, \mathbf{w}_0\right),\label{59}\\
			~~~~&~\text{(\ref{29})}, \text{(\ref{54})}\label{60},
		\end{align}
		where $\mathbf{a}_1=\begin{bmatrix}
			\mu_1&\mu_2\\
			\mu_2&\mu_3
		\end{bmatrix}$, $\mathbf{a}_2=\begin{bmatrix}
			\mu_4\\
			\mu_5
		\end{bmatrix}$, and $\mathbf{a}_3=\mu_6$. According to Proposition \ref{3.1},  $f_j\left(\mathbf{w}_c, \mathbf{w}_0\right)$ in (\ref{59}) is expressed as $f_j\left(\mathbf{w}_c, \mathbf{w}_0\right)=\operatorname{Tr}\left(\boldsymbol{\zeta}_{j}\boldsymbol{\Sigma}\right)$, $\forall j\in\boldsymbol{\mathcal{J}}=\left\{1,2,...,6\right\}$. Let $\mathbf{S}\left(r, \mathbf{w}_0, \mathbf{w}_c\right)=2\Re\left(r\mathbf{h}_c^{H}\mathbf{w}_c\right)-\left|r\right|^2\left(\left|\mathbf{h}_c^{H}\mathbf{w}_0\right|^2+\sigma_{n}^2\right)$.
		Then, by adding non-convex constraint (\ref{59}) as a penalty term to the 
		\vspace{-0.2cm}
		objective function (\ref{57}),  Problem (P4.7) is reformulated as
		\vspace{0.1cm}
		\begin{align}
			\left(\text{P4.8}\right)\mathop{\text{min}}\limits_{\left\{r,\boldsymbol{\Omega}, \mathbf{w}_c,\mathbf{w}_0, \boldsymbol{\mu}\right\}} &~\frac{1}{2\kappa}F\left(\mathbf{w}_c, \mathbf{w}_0, \boldsymbol{\mu}\right)-\mathbf{S}\left(r, \mathbf{w}_0, \mathbf{w}_c\right)\label{61}\\
			\text{s.t.} 
			~~~~~&~\text{(\ref{29})}, \text{(\ref{54})}, \text{(\ref{58})}\label{62},
		\end{align}   
		where $F_1\left(r, \mathbf{w}_c, \mathbf{w}_0, \boldsymbol{\mu}\right)=\frac{1}{2\kappa}F\left(\mathbf{w}_c, \mathbf{w}_0, \boldsymbol{\mu}\right)-\mathbf{S}\left(r, \mathbf{w}_0, \mathbf{w}_c\right)$,  $\kappa\textgreater0$ denotes the penalty parameter, and $F\left(\mathbf{w}_c,\mathbf{w}_0, \boldsymbol{\mu}\right)=\sum_{j=1}^{6}\left|f_j\left(\mathbf{w}_c,\mathbf{w}_0\right)-\mu_j\right|^2$. To address Problem (P4.8), we decompose it as two subproblems. Specifically, with fixed  beamforming vectors $\mathbf{w}_c$ and $\mathbf{w}_0$,  the subproblem to acquire $\boldsymbol{\Omega}$ and $\boldsymbol{\mu}$ is written as
		\begin{align}
			\left(\text{P4.9}\right)\mathop{\text{min}}\limits_{\left\{\boldsymbol{\Omega},  \boldsymbol{\mu}\right\}}~~~~&~\frac{1}{2\kappa}F\left(\mathbf{w}_c,\mathbf{w}_0,\boldsymbol{\mu}\right)\label{63}\\
			\text{s.t.} 
			~~~~~&~\text{(\ref{54})}, \text{(\ref{58})}\label{64},
		\end{align} 
		which is a semi-definite programming problem (SDP) and can be addressed by CVX \cite{cvx_software}. Then, with fixed $\boldsymbol{\Omega}$ and $\boldsymbol{\mu}$, the subproblem updating $\mathbf{w}_0$ and $\mathbf{w}_c$ is reconstructed as
		\begin{align}
			\left(\text{P4.10}\right)\mathop{\text{min}}\limits_{\left\{r, \mathbf{w}_0,  \mathbf{w}_c\right\}}&~F_1\left(r, \mathbf{w}_c, \mathbf{w}_0, \boldsymbol{\mu}\right)\label{631}\\
			\text{s.t.} 
			~~~&~\text{(\ref{29})}\label{64},
		\end{align} 
		where constraint (\ref{29}) and the  part $-\mathbf{S}\left(r, \mathbf{w}_0, \mathbf{w}_c\right)$ in (\ref{631})  are convex, while the residual part  $F\left(\mathbf{w}_c,\mathbf{w}_0,\boldsymbol{\mu}\right)$ is  non-convex  with respect to $\mathbf{w}_c$ and $\mathbf{w}_0$. To address this problem, we have the following lemma.  
		\begin{algorithm}[!t]
			\caption{AO-Based Algorithm for Problem (P3.1) }
			\label{Algorithm 1}
			\begin{algorithmic}[1]
				\State	 {\textbf{Input:} $P_T$,  $R_s\in\mathcal{R}_s=\left\{R_{s,1},...,R_{s,G}\right\}$, $\mathbf{h}_c$ and $\mathbf{H}_0$,  $\iota$}
				\State	 {\textbf{Output:} $R_{s}^{\star}$,   $R_{c}^{\star}$, $\mathbf{w}_{c}^{\star}$, $\mathbf{w}_{0}^{\star}$; }
				\State	 {\textbf{Initialize:}   $\mathbf{w}_c^{'}$, $\mathbf{w}_0^{'}$, $R_c^{'}$,   $\kappa^{\left(0\right)}$,    $i=1$, $q=0$;    }
				\For {$g=1:G$}
				\State \textbf{while} no convergence  \textbf{do}
				\State \quad \textbf{while} no convergence   \textbf{do}
				\State{\quad\quad Obtain $R_{c,g}^{\left(i\right)}$, $\hat{\rho}_{b,g}^{\left(i\right)}$, $\nu_{g}^{\left(i\right)}$ by solving Problem (P4.4);}
				\State{\quad\quad Obtain $\boldsymbol{\Omega}_{g}^{\left(i\right)}$, $\boldsymbol{\mu}_{g}^{\left(i\right)}$ by solving Problem (P4.9) };
				\State {\quad\quad Obtain $r_{g}^{\left(i\right)}$,  $\mathbf{w}_{c,g}^{\left(i\right)}$ and $\mathbf{w}_{0,g}^{\left(i\right)}$ by solving Problem (P4.11); }
				\State{\quad\quad Update $b_1^{\left(i\right)}=F_1^{\left(i\right)}\left(r, \mathbf{w}_c,\mathbf{w}_0, \boldsymbol{\mu} \right)$};
				\State{\quad\quad Update $\mathbf{b}_2^{\left(i\right)}=\left[b_{2,1}^{\left(i\right)}, b_{2,2}^{\left(i\right)}, \cdots, b_{2,6}^{\left(i\right)}\right]$ with $b_{2,j}^{\left(i\right)}=\left|f_j^{\left(i\right)}\left(\mathbf{w}_{c,g},\mathbf{w}_{0,g}\right)-\mu_{g,j}^{\left(i\right)}\right|^2$, for $j\in\boldsymbol{\mathcal{J}}$ };
				\State {\quad\quad Update $D_{o,g}^{\left(i\right)}$,  $\mathbf{w}_c^{'}=\mathbf{w}_{c,g}^{\left(i\right)}$, $\mathbf{w}_0^{'}=\mathbf{w}_{0,g}^{\left(i\right)}$, $R_c^{'}=R_{c,g}^{\left(i\right)}$;  }
				\State {\quad\quad Update $i=i+1$;}
				\State \quad\textbf{end while}
				\State {\quad Update $q=q+1$};
				\State {\quad Update $\kappa^{\left(q\right)}=\iota\kappa^{\left(q-1\right)}$};
				\State  \textbf{end while}  
				\EndFor
				\State	{$g^{\star}=\mathop{\text{arg}\text{min}}\limits_{g}  \left\{D_{o,1}, \cdots, D_{o,G} \right\}$, $R_s^{\star}=R_{s,g^{\star}}$, $R_c^{\star}=R_{c,g^{\star}}$, $\mathbf{w}_c^{\star}=\mathbf{w}_{c,g^{\star}}$, $\mathbf{w}_o^{\star}=\mathbf{w}_{0,g^{\star}}$.}
			\end{algorithmic}
		\end{algorithm}
		\begin{Lemma}
			By expanding the functional expression and applying the SCA method, the non-convex part $\left|f_j\left(\mathbf{w}_c,\mathbf{w}_0\right)-\mu_j\right|^2$ in $F\left(\mathbf{w}_c,\mathbf{w}_0,\boldsymbol{\mu}\right)$ is transformed into a  convex version given in (\ref{621}),
			where $G_j\left(\mathbf{w}_c, \mathbf{w}_0\right)$ denotes the  convex surrogate obtained via SCA, with  $\mathbf{\Lambda}_j=\mu_{j}^{*}\boldsymbol{\zeta}_j+\mu_{j}\left(\boldsymbol{\zeta}_j\right)^{H}$ , $\mathbf{a}_{1,j}^{'}=\left(\left(h^{'}\right)^{*}\left(\boldsymbol{\zeta}_j\right)\mathbf{w}_0^{'}+h^{'}\left(\boldsymbol{\zeta}_j\right)^H\mathbf{w}_0^{'}\right)^H$,  $h^{'}=\left(\mathbf{w}_c^{'}\right)^H\boldsymbol{\zeta}_{j}\mathbf{w}_c^{'}$, $\mathbf{a}_{2,j}^{'}=\left(\left(g^{'}\right)^{*}\left(\boldsymbol{\zeta}_j\right)\mathbf{w}_c^{'}+g^{'}\left(\boldsymbol{\zeta}_j\right)^H\mathbf{w}_c^{'}\right)^H$,  and $g^{'}=\left(\mathbf{w}_0^{'}\right)^H\boldsymbol{\zeta}_{j}\mathbf{w}_0^{'}$. Besides, $\mathbf{w}_c^{'}$ and $\mathbf{w}_0^{'}$ denote the values obtained in the previous iteration, respectively;
		\end{Lemma}
		\vspace{0.1cm}
		\begin{proof}
			Please see Appendix D. \hfill $\blacksquare$
		\end{proof}
		\vspace{0.1cm}
		
		Based on Lemma 4.3, Problem (P4.10) is transformed into the following convex problem, i.e.,
		\setcounter{equation}{62}
		\begin{align}\label{68}
			\left(\text{P4.11}\right)\mathop{\text{min}}\limits_{\left\{r,\mathbf{w}_0,\mathbf{w}_c  \right\}}~~~~&~F_2\left(r, \mathbf{w}_c,\mathbf{w}_0\right),\\
			\text{s.t.} 
			~~~~~~~&~\text{(\ref{29})},
		\end{align}	
		where  
		$F_2\left(r, \mathbf{w}_c,\mathbf{w}_0\right)=\sum_{j=1}^{6}\frac{1}{2\kappa}G_j\left(\mathbf{w}_c, \mathbf{w}_0\right)-\mathbf{S}\left(r, \mathbf{w}_0, \mathbf{w}_c\right)$.
		Obviously, Problem (P4.11) is a convex problem, which is solved by existing convex techniques \cite{cvx_software}. 
		The overall algorithm is summarized in Algorithm \uppercase\expandafter{\romannumeral1}, where  $\iota$ denotes the step size to control the decreasing speed of $\kappa$ and  the subscript $g$ is adopted to denote the optimization parameters of the $g$-th model.  Moreover, the convergence criterion for the inner layer and outer layer are defined as $\left|b_1^{\left(i\right)}-b_1^{\left(i-1\right)}\right|\leq 10^{-3}$ and $\text{max}\left\{b_{2,j}^{\left(i\right)}-b_{2,j}^{\left(i-1\right)}, j\in\boldsymbol{\mathcal{J}} \right\}\leq 10^{-6}$, respectively. 
		\section{Simulation Results}
		In this section, we provide simulation results to evaluate the performance of the proposed SA-ASCC scheme for the considered bi-static ISSC system. The  experimental setup is given as follows:
		\begin{itemize}
			\item \textbf{Datasets}: Experimental validation of the SA-ASCC scheme is conducted on the CUB-200-2011 dataset, which consists of 11,788 images covering 200 bird categories. We split the dataset into 5,994 training samples and 5,794 testing samples. All input images are resized to dimensions of 256$\times$256 for consistency.
			\item \textbf{DNN model architecture}: We leverage the hyperprior-based DNN framework to serve as the backbone for both the feature extraction function $\boldsymbol{f}_{\varphi}$ and the feature recovery function $\boldsymbol{g}_{\nu}$. We jointly train $G$  hyperprior-based DNN models, with source coding rate ranging from $1.3\times 10^3$ to $5.7\times10^4$.
			\item \textbf{Parameter settings}: We consider the ISSC system with one TR equipped with $N_T=8$ antennas, one target, one single-antenna CRE and one SRE with $N_R=4$ antennas. The block length is  set as $L=256$. The positions of the TR, target, SRE and CRE are set as $\mathbf{p}_b=\left[0,0\right], \mathbf{p}_o=\left[100,50\right], \mathbf{p}_r=\left[0,50\right], \mathbf{p}_c=\left[20,30\right]$, respectively. The bandwidth of the considered ISSC system is set as $B=100$MHz. The path loss exponents for the TR-CRE and TR-target links are set as $\epsilon_1 =3$ and $\epsilon_2=2$, respectively\cite{11240213}. The reflection factor is set as $\varrho=0.6$. The  power of CSCG noise is set as $\sigma_n^2=\sigma_s^2=-50 $ dBm. The TS error covariance is set as $\sigma_{T}=100$ ns.  The number of samples is set as $N=1024$. The HCRB threshold is set as $\Pi=0.01$. The number of DNN models is $G=20$. The step size is set as $\iota=0.8$. 
			\item\textbf{Baseline schemes:} To evaluate the proposed SA-ASCC scheme for ISSC systems, we consider two representative benchmarks: DJSCC\cite{8723589} based on the continuous JSCC architecture and BPG based on the conventional  SSCC architecture. Specifically, for the DJSCC baseline, the input images are directly encoded into analog channel symbols, and the model is trained under the same SNR and ABR settings as those used in the corresponding test scenarios. For both benchmarks, water-filling (WF) is employed for power allocation according to the channel conditions, while zero-forcing (ZF) is adopted for beamforming to suppress inter-user interference. For brevity, the comparison benchmark schemes are denoted  as DJSCC-WF-ZF and BPG-WF-ZF, respectively.
		\end{itemize}
		\subsection{Image Reconstruction Performance}
		\begin{figure}[!t]
			\setlength{\belowcaptionskip}{-5pt}
			\setlength{\abovecaptionskip}{0.2cm}
			\centering
			\vspace{-1em}
			\includegraphics[width=2.3in]{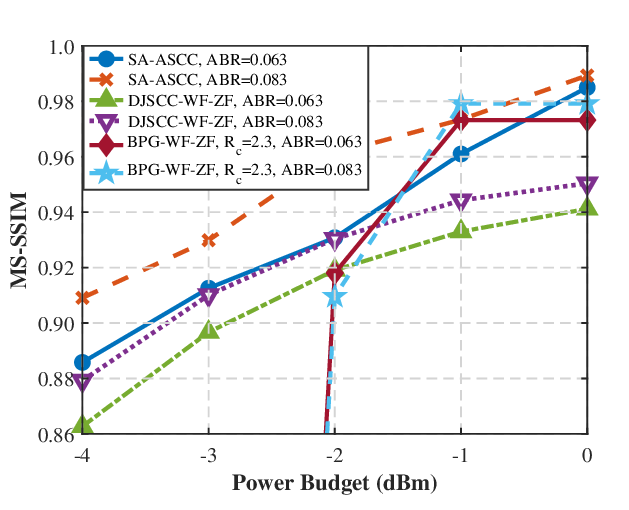}
			\caption{MS-SSIM comparison under different $P_T$.}
			\label{fig:psnr}
		\end{figure}
		\begin{figure}[!t]
			\setlength{\belowcaptionskip}{-10pt}
			\setlength{\abovecaptionskip}{0.2cm}
			\centering
			\vspace{-1em}
			\includegraphics[width=2.3in]{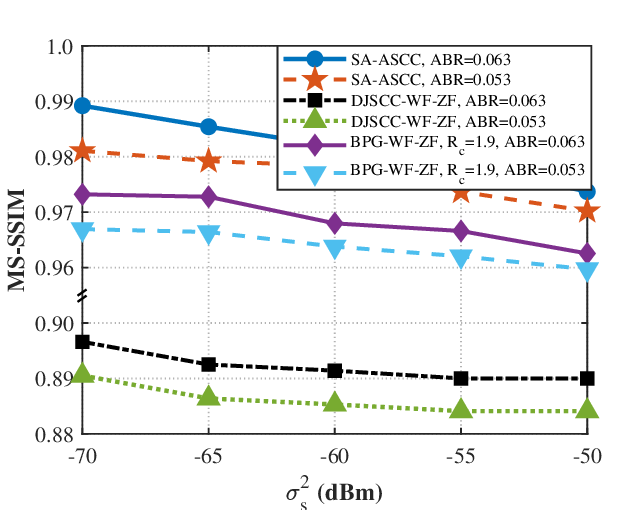}
			\caption{MS-SSIM comparison under different $\sigma_s^2$.}
			\label{fig:crb}
		\end{figure}
		
		In this section, multi-scale structural similarity (MS-SSIM) is selected as the performance metric of SemCom\cite{yuan2025adaptivesourcechannelcodingmultiuser}. As shown in Fig.~\ref{fig:psnr}, the proposed SA-ASCC scheme exhibits superior overall MS-SSIM performance compared with the DJSCC-WF-ZF and BPG-WF-ZF benchmarks across different $P_T$ and ABR settings, with $\Pi=0.06$ . This performance gain is primarily attributed to its adaptive source-channel coding rate mechanism, which dynamically adjusts the source-channel coding rates according to the varying channel conditions, i.e., SINR, and consequently enables SA-ASCC to effectively mitigate the cliff effect suffered by the conventional BPG-WF-ZF scheme. Moreover, when $P_T=-4~\mathrm{dBm}$, corresponding to an SINR of $-3.2~\mathrm{dB}$, the proposed SA-ASCC scheme with $\mathrm{ABR}=0.063$ achieves a 2.68$\%$ MS-SSIM improvement over DJSCC-WF-ZF with the same ABR, demonstrating the superiority of the proposed SA-ASCC  in low-SINR regimes. When $P_T=-2~\mathrm{dBm}$, SA-ASCC with $\mathrm{ABR}=0.063$ achieves comparable MS-SSIM performance to DJSCC-WF-ZF with $\mathrm{ABR}=0.083$, while reducing the required bandwidth by 24.1$\%$, further demonstrating the bandwidth efficiency of the proposed SA-ASCC scheme.

Fig.~\ref{fig:crb} illustrates the impact of sensing noise under $P_T=-3\mathrm{dBm}$ and $\Pi=0.065$. The simulation results show that the MS-SSIM performance degrades as the sensing noise increases. This is because, to satisfy the same sensing accuracy constraint under a higher sensing noise level, the proposed SA-ASCC scheme allocates more transmission resources to the sensing task, thereby reducing the transmission resources available for semantic transmission and consequently degrading the MS-SSIM performance.

		\subsection{Target Positioning Performance}
		\begin{figure}[!t]
			\setlength{\belowcaptionskip}{-5pt}
			\setlength{\abovecaptionskip}{0.2cm}
			\centering
			\vspace{-0.5em}
			\includegraphics[width=2.3in]{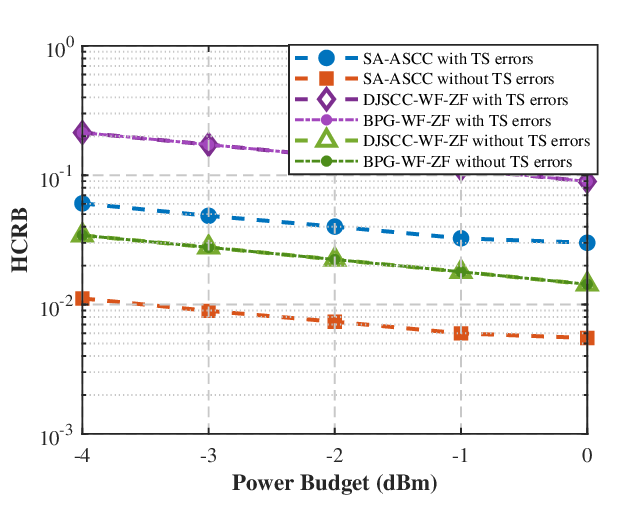}
			\caption{HCRB comparison under different $P_T$.}
			\label{fig:5}
		\end{figure}
		
		As shown in Fig.~\ref{fig:5}, the proposed SA-ASCC scheme achieves a consistently lower HCRB than the DJSCC-WF-ZF and BPG-WF-ZF benchmarks across different $P_T$. This sensing performance gain is mainly attributed to the superior SemCom performance of SA-ASCC in the low-SINR regime, which enables more transmission resources to be allocated to sensing while maintaining the required communication performance. In contrast, the DJSCC-WF-ZF and BPG-WF-ZF schemes are constrained by the limited degrees of freedom (DoFs) caused by the stringent interference-suppression requirements of ZF beamforming. Moreover, the DJSCC-WF-ZF scheme and the BPG-WF-ZF scheme exhibit identical HCRB performance since they employ the same power allocation and beamforming strategies under the same system settings. Furthermore, the pronounced HCRB degradation in the presence of TS errors highlights the importance of accounting for imperfect TS in practical system design.
		\begin{figure}[!t]
			\setlength{\belowcaptionskip}{-5pt}
			\setlength{\abovecaptionskip}{0.2cm}
			\centering
			\vspace{-1em}
			\includegraphics[width=2.3in]{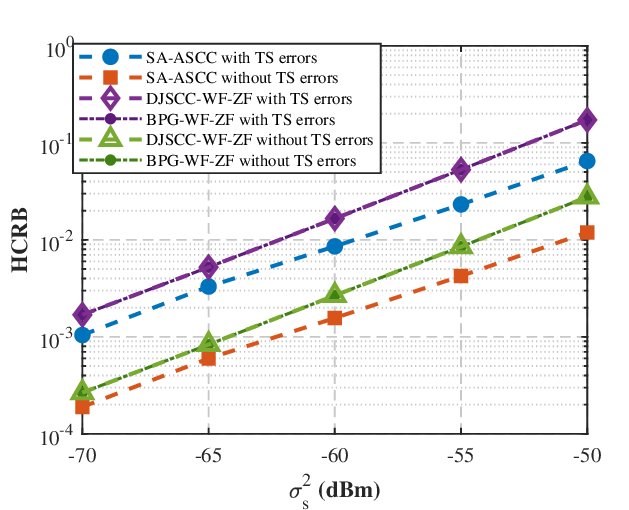}
			\caption{HCRB comparison under different $\sigma_s^2$.}
			\label{fig:6}
		\end{figure}

Fig.~\ref{fig:6} shows the HCRB performance under different sensing noise levels. It is observed that the HCRB monotonically increases with the sensing noise level, since stronger sensing noise reduces the quality of the received sensing signals and consequently degrades the sensing estimation accuracy. Furthermore, a comprehensive analysis of Fig.~\ref{fig:5} and Fig.~\ref{fig:6} reveals that the proposed SA-ASCC scheme consistently achieves a lower HCRB than the DJSCC-WF-ZF and BPG-WF-ZF benchmarks, regardless of whether the system operates under perfect TS or suffers from TS errors.
		\subsection{Tradeoff  and Power Allocation Analysis}

Fig.~\ref{fig:8} shows the achievable  performance region under different TS error levels, with $P_T=0$ dBm. It is observed that the MS-SSIM increases with the HCRB and gradually approaches saturation, since higher HCRB  allows more transmit power to be allocated to SemCom. Moreover, for a given HCRB, the SA-ASCC scheme with $\sigma_T=1$ ns achieves a higher MS-SSIM than SA-ASCC scheme with $\sigma_T=100$ ns. This performance gap arises because a larger TS error requires more transmit power to be allocated to sensing to achieve the same sensing accuracy, thereby reducing the power available for semantic transmission and degrading the MS-SSIM performance. Consequently, the achievable  performance region decreases as the TS error increases. Furthermore, compared with the DJSCC-WF-ZF and BPG-WF-ZF schemes, the proposed SA-ASCC scheme achieves a larger achievable  performance region, demonstrating its superiority in jointly supporting SemCom and sensing.		

		\begin{figure}[!t]
			\setlength{\belowcaptionskip}{-5pt}
			\setlength{\abovecaptionskip}{0.2cm}
			\centering
			\vspace{-1em}
			\includegraphics[width=2.3in]{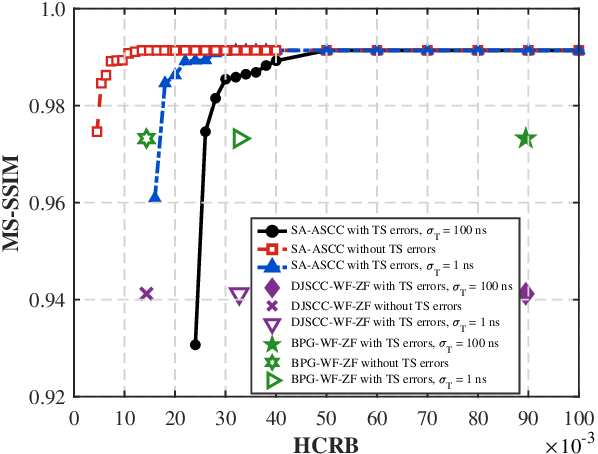}
			\caption{MS-SSIM comparison versus HCRB.}
			\label{fig:8}
		\end{figure} 
				\begin{figure}[!t]
	\setlength{\belowcaptionskip}{-5pt}
	\setlength{\abovecaptionskip}{0.2cm}
	\centering
	\vspace{-1em}
	\includegraphics[width=2.3in]{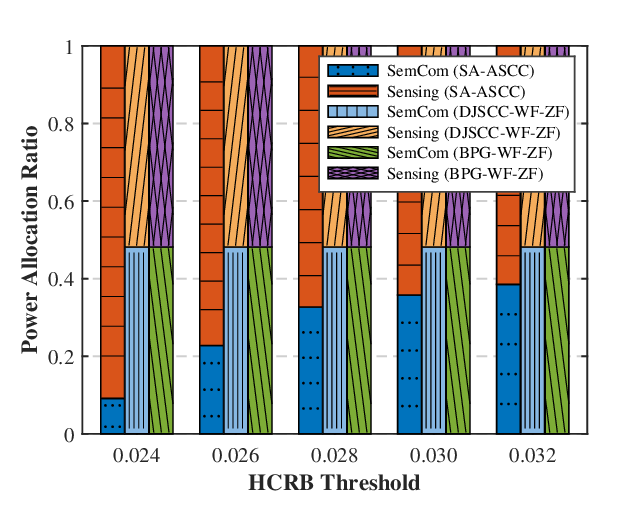}
	\caption{power allocation versus sensing accuracy requirement.}
	\label{p2}
\end{figure}

Fig.~\ref{p2} illustrates the power allocation between SemCom and sensing under different HCRB thresholds, with $P_T=0$ dBm and $\sigma_{T}=100\text{ns}$. As the HCRB threshold increases, i.e., as the sensing accuracy requirement becomes less stringent, the proposed SA-ASCC scheme progressively allocates more power resources to SemCom. This trend arises because relaxing the HCRB constraint reduces the transmission resources required to satisfy the sensing accuracy requirement, thereby allowing more power to be allocated to semantic transmission. Notably, the SA-ASCC scheme allocates more power resources to sensing than the DJSCC-WF-ZF and BPG-WF-ZF schemes, which accounts for its superior sensing performance observed in the preceding results. Moreover, the power allocation of the DJSCC-WF-ZF and BPG-WF-ZF schemes remains unchanged with the HCRB threshold, as their WF-based power allocation is determined by the channel conditions and cannot adapt to varying sensing accuracy requirements.
		\section{Conclusion}
		This paper proposed an integrated transmission framework composed of  SA-ASCC and beamforming design for bi-static ISSC system. We first derived  the E2E distortion as a function of the BER for any fixed source coding rate, and then  derived the HCRB for target position under imperfect TS. To characterize the achievable region between SemCom and sensing performance,  an E2E distortion minimization problem was  formulated   by considering the HCRB threshold, channel uses,  and power budget. To address this problem,  an  AO algorithm is proposed to decompose it into two subproblems, which are solved by the exhaustive search method and a combination of FP and SCA, respectively. Finally, simulation results demonstrated that the proposed  SA-ASCC scheme achieved superior performance   compared with the DJSCC‑WF‑ZF and BPG-WF-ZF benchmarks.
		\appendices
		\section{Derivation of HCRB with TS errors}
		The $\left(i,j\right)$-th element of $\mathbf{J}_D$ is derived as
		\begin{align}\label{69}
			J_{D}^{i,j}&=-\mathbb{E}\left(\frac{\partial^{2}\text{ln}\hat{\boldsymbol{f}}\left(\breve{\mathbf{y}}_{s}^{'}|\boldsymbol{\chi}\right)}{\partial\boldsymbol{\chi}^{\left(i\right)}\partial\boldsymbol{\chi}^{\left(j\right)}}\right)\\
			&=2\Re\left\{\frac{\partial\widehat{\mathbf{x}}^{H}}{\partial\boldsymbol{\chi}^{\left(i\right)}}\boldsymbol{\Xi}^{-1}\frac{\partial\widehat{\mathbf{x}}}{\partial\boldsymbol{\chi}^{\left(j\right)}}\right\}\label{651}\\
			&~~~+\operatorname{Tr}\left(\boldsymbol{\Xi}^{-1}\frac{\partial\boldsymbol{\Xi}^{H}}{\partial\boldsymbol{\chi}^{\left(i\right)}}\boldsymbol{\Xi}^{-1}\frac{\partial\boldsymbol{\Xi}}{\partial\boldsymbol{\chi}^{\left(j\right)}}\right)\nonumber\\
			&=\frac{2}{\sigma_{s}^2}\Re\left\{\frac{\partial\widehat{\mathbf{x}}^{H}}{\partial\boldsymbol{\chi}^{\left(i\right)}}\frac{\partial\widehat{\mathbf{x}}}{\partial\boldsymbol{\chi}^{\left(j\right)}}\right\}\label{6611},
		\end{align}
		where $\boldsymbol{\chi}^{\left(i\right)}$ and $\boldsymbol{\chi}^{\left(j\right)}$ denote the $i$-th   and $j$-th element of $\boldsymbol{\chi}$, respectively;  (\ref{69}) is obtained based on (\ref{18}); (\ref{651}) is obtained according to the Slepian-Bangs Formula in \cite{10050406}; (\ref{6611}) is obtained due to $\frac{\partial\boldsymbol{\Xi}}{\partial\boldsymbol{\chi}^{\left(j\right)}}=0$ , for $\forall j\in \mathbb{Z}^+$.  Furthermore, based on (\ref{15}), $\frac{\partial\widehat{\mathbf{x}}}{\partial o_x}$ is derived as
		\begin{equation}\label{70}
			\frac{\partial\widehat{\mathbf{x}}}{\partial{o}_{x}}=\text{vec}\left(\mathbf{H}_{0}^{'}\widetilde{\mathbf{X}}+\tau^{'}\mathbf{H}_{0}\overline{\mathbf{X}}\right),
		\end{equation}
		where $\overline{\mathbf{X}}$ is defined in (\ref{661}).
		\begin{figure*}[b] 
			\hrulefill  
			\vspace{-0.5em}  
			\begin{equation}\label{661}  
				\overline{\mathbf{X}}=\left[-jw_1e^{-jw_1\tau}\widetilde{\mathbf{x}}\left(w_1\right),-jw_2e^{-jw_2\tau}\widetilde{\mathbf{x}}\left(w_2\right),...,-jw_Ne^{-jw_N\tau}\widetilde{\mathbf{x}}\left(w_N\right)\right].~~~~~~~~~~~~~~~~~~~~~~~~~~~~~~~~~~~~~~~~~
			\end{equation}
			\begin{equation}\label{75}
				J_{D}^{o_{x}o_{x}}\stackrel{(a)}{=}\frac{2}{\sigma_{s}^2}\Re\left\{\frac{\partial\widehat{\mathbf{x}}^{H}}{\partial o_x}\frac{\partial\widehat{\mathbf{x}}}{\partial o_x}\right\} \stackrel{(b)}{=}\frac{2}{\sigma_{s}^2}\Re\left\{\text{vec}^{H}\left(\mathbf{H}_{0}^{'}\widetilde{\mathbf{X}}+\tau^{'}\mathbf{H}_{0}\overline{\mathbf{X}}\right)\cdot\text{vec}\left(\mathbf{H}_{0}^{'}\widetilde{\mathbf{X}}+\tau^{'}\mathbf{H}_{0}\overline{\mathbf{X}}\right)\right\}
				\stackrel{(c)}{=}\frac{2}{\sigma_{s}^2}\Re\left\{\operatorname{Tr}\left(\boldsymbol{\zeta}_{1}\boldsymbol{\Sigma}\right)\right\},\tag{72}
			\end{equation}
			\vspace{-3em}
		\end{figure*} 
		$\tau'$ and $\mathbf{H}_{0}^{'}$  denote the first-order derivatives of $\tau$ and $\mathbf{H}_0$ with respect to   $o_x$, respectively,  with $\mathbf{H}_{0}^{'}=\beta_{0}^{'}\boldsymbol{\alpha}\left(\phi\right)\boldsymbol{\alpha}\left(\theta\right)^{H}+\beta_{0}\boldsymbol{\alpha}^{'}\left(\phi\right)\boldsymbol{\alpha}^{H}\left(\theta\right)+\beta_{0}\boldsymbol{\alpha}\left(\phi\right)\left(\boldsymbol{\alpha}^{'}\left(\theta\right)\right)^{H}$ and $\tau^{'}=\frac{1}{c}\Upsilon_1$.
		Moreover, $\beta_{0}^{'}=	-\frac{\epsilon_2\varrho}
		{2\left(d_{b,o}+d_{r,o}\right)^{\epsilon_2/2+1}}
		\Upsilon_1$, with $\Upsilon_1=\left(\frac{\left(o_x-b_x\right)}{d_{b,o}}+\frac{\left(o_x-r_x\right)}{d_{r,o}}\right)$. Here,  $d_{b,o}$ and $d_{r,o}$ denote the distance between the target and the ISSC TR and the distance between the target and the  SRE, respectively, with $d_{b,o}=\sqrt{\left(o_x-b_x\right)^2+\left(o_y-b_y\right)^2}$ and $d_{r,o}=\sqrt{\left(o_x-r_x\right)^2+\left(o_y-r_y\right)^2}$. Moreover, $\boldsymbol{\alpha}^{'}\left(\phi\right)$ and $\boldsymbol{\alpha}^{'}\left(\theta\right)$ denote the first-order derivatives of $\boldsymbol{\alpha}\left(\phi\right)$ and $\boldsymbol{\alpha}\left(\theta\right)$  with respect to $o_x$, respectively, where $\boldsymbol{\alpha}^{'}\left(\phi\right)=\phi^{'}\tilde{\boldsymbol{\alpha}}\left(\phi\right)\odot\boldsymbol{\alpha}\left(\phi\right)$ and $\boldsymbol{\alpha}^{'}\left(\theta\right)=\theta^{'}\tilde{\boldsymbol{\alpha}}\left(\theta\right)\odot\boldsymbol{\alpha}\left(\theta\right)$,
		with $\phi^{'}=\frac{r_y-o_y}{d_{r,o}^2}$, $\tilde{\boldsymbol{\alpha}}\left(\phi\right)=\left[0,j\pi\text{cos}\left(\phi\right),...,j\left(
		N_R-1\right)\pi\text{
			cos}\left(\phi\right)\right]^{T}$, $\theta^{'}=\frac{
			b_y-o_y}{d_{b,o}^2}$, and $\tilde{\boldsymbol{\alpha}}\left(\theta\right)=\left[0,j\pi\text{cos}\left(\theta\right),...,j\left(
		N_T-1\right)\pi\text{
			cos}\left(\theta\right)\right]^{T}$. Similarly, $\frac{\partial\widehat{\mathbf{x}}}{\partial o_y}$ is derived as
		\begin{equation}\label{72}
			\frac{\partial\widehat{\mathbf{x}}}{\partial o_y}=\text{vec}\left(\mathbf{H}_{0}^{''}\widetilde{\mathbf{X}}+\tau^{''}\mathbf{H}_{0}\overline{\mathbf{X}}\right),
		\end{equation}
		where $\mathbf{H}_{0}^{''}$ and $\tau^{''}$ respectively denote the first-order derivatives of $\mathbf{H}_0$ and $\tau$ with respect to $o_y$. Specifically,  $\mathbf{H}_{0}^{''}=\beta_{0}^{''}\boldsymbol{\alpha}\left(\phi\right)\boldsymbol{\alpha}\left(\theta\right)^{H}+\beta_{0}\boldsymbol{\alpha}^{''}\left(\phi\right)\boldsymbol{\alpha}^{H}\left(\theta\right)+\beta_{0}\boldsymbol{\alpha}\left(\phi\right)\left(\boldsymbol{\alpha}^{''}\left(\theta\right)\right)^{H}$, and $\tau^{''}=\frac{1}{c}\Upsilon_{2}$,
		where $\beta_{0}^{''}=-\frac{\epsilon_2\varrho}
		{2\left(d_{b,o}+d_{r,o}\right)^{\epsilon_2/2+1}}
		\Upsilon_2$, with $\Upsilon_2=\left(\frac{\left(o_y-b_y\right)}{d_{b,o}}+\frac{\left(o_y-r_y\right)}{d_{r,o}}\right)$. Moreover,
		$\boldsymbol{\alpha}^{''}\left(\phi\right)$ and $\boldsymbol{\alpha}^{''}\left(\theta\right)$ denote the first-order derivatives of $\boldsymbol{\alpha}\left(\phi\right)$ and $\boldsymbol{\alpha}\left(\theta\right)$ with respect to $o_y$, respectively, where $\boldsymbol{\alpha}^{''}\left(\phi\right)=\phi^{''}\tilde{\boldsymbol{\alpha}}\left(\phi\right)\odot\boldsymbol{\alpha}\left(\phi\right)$ and $\boldsymbol{\alpha}^{''}\left(\theta\right)=\theta^{''}\tilde{\boldsymbol{\alpha}}\left(\theta\right)\odot\boldsymbol{\alpha}\left(\theta\right)$, with
		$\phi^{''}=\frac{o_x-r_x}{d_{r,o}^2}$ and  $\theta^{''}=\frac{
			o_x-b_x}{d_{b,o}^2}$. Then, the  first-order derivative of $\widehat{\mathbf{x}}$ with respect to $\Delta \tau_{t,r}$ is derived as
		\begin{equation}\label{74}
			\frac{\partial\widehat{\mathbf{x}}}{\partial\Delta \tau_{t,r}}=\text{vec}\left(\mathbf{H}_{0}\overline{\mathbf{X}}\right).
		\end{equation}
		$\mathbf{J}_{D}^{c_1}$ is expressed as $\mathbf{J}_{D}^{c_1}=\begin{bmatrix}
			J_{D}^{o_x o_x}&J_{D}^{o_x o_y}\\
			J_{D}^{o_x o_y}&J_{D}^{o_y o_y}
		\end{bmatrix}$. 
		Then, $\mathbf{J}_{D}^{o_{x}o_{x}}$ is derived as (\ref{75}), where equality $(a)$ is obtained based on (\ref{6611}); equality $(b)$ is obtained by substituting (\ref{70}) into (\ref{6611}); 
		equality $(c)$ is obtained by utilizing the property that  $\text{vec}^{H}\left(A\right)\text{vec}\left(B\right)=\operatorname{Tr}\left(A^{H}B\right)$, where $\boldsymbol{\zeta}_{1}=\sum_{n=1}^{N}\left(\dot{\mathbf{H}}_{0}^{n}\right)^H\cdot\dot{\mathbf{H}}_{0}^{n}$, with $\dot{\mathbf{H}}_{0}^{n}=\mathbf{H}_{0}^{'}-jw_{n}\tau^{'}\mathbf{H}_{0}$. Similar to the process of deriving (\ref{75}), $J_{D}^{o_x o_y}$ and $J_{D}^{o_y o_y}$ are derived as the following form based on (\ref{70}) and (\ref{72}), i.e.,
		\setcounter{equation}{72}
		\begin{align}
			J_{D}^{o_x o_y}=\frac{2}{\sigma_{s}^2}\Re\left\{\operatorname{Tr}\left(\boldsymbol{\zeta}_{2}\boldsymbol{\Sigma}\right)\right\}\label{721},\\
			J_{D}^{o_y o_y}=\frac{2}{\sigma_{s}^2}\Re\left\{\operatorname{Tr}\left(\boldsymbol{\zeta}_{3}\boldsymbol{\Sigma}\right)\right\},\label{731}
		\end{align}  
		where $\boldsymbol{\zeta}_{2}=\sum_{n=1}^{N}\left(\dot{\mathbf{H}_{0}^n}\right)^{H}\ddot{\mathbf{
				H}}_0^n$ and $\boldsymbol{\zeta}_{3}=\sum_{n=1}^{N}\left(\ddot{\mathbf{
				H}}_0^n\right)^{H}\cdot \ddot{\mathbf{
				H}}_0^n$, with $\ddot{\mathbf{
				H}}_0^n=\mathbf{H}_{0}^{''}-jw_{n}\tau^{''}\mathbf{H}_0$. Based on (\ref{70})-(\ref{74}), the elements of $J_D^{\mathbf{c}_2}=\left[J_{D}^{o_x\Delta \tau_{t,r}}, J_{D}^{o_y\Delta \tau_{t,r}}\right]^T$ and $J_{D}^{c_3}=J_{D}^{\Delta\tau_{t,r}\Delta\tau_{t,r}}$ are derived as
		\begin{align}
			J_{D}^{o_x\Delta \tau_{t,r}}
			&=\frac{2}{\sigma_{s}^2}\Re\left\{\operatorname{Tr}\left(\boldsymbol{\zeta}_{4}\boldsymbol{\Sigma}\right)\right\}\label{77},\\
			J_{D}^{o_y\Delta \operatorname{Tr}}&=\frac{2}{\sigma_{s}^2}\Re\left\{\operatorname{Tr}\left(\boldsymbol{\zeta}_{5}\boldsymbol{\Sigma}\right)\right\}\label{771},\\
			J_{D}^{\Delta\tau_{t,r}\Delta\tau_{t,r}}&=\frac{2}{\sigma_s^{2}}\Re\left\{\operatorname{Tr}\left(\boldsymbol{\zeta}_{6}\boldsymbol{\Sigma}\right)\right\},\label{781}
		\end{align}
		where (\ref{77}), (\ref{771}) and (\ref{781}) are obtained similarly to the process of obtaining (\ref{721}) and (\ref{731}). 
		Here, $\boldsymbol{\zeta}_{4}=\sum_{n=1}^{N}-jw_n\left(\dot{\mathbf{H}_0^n}\right)^H\mathbf{H}_0$, $\boldsymbol{\zeta}_{5}=\sum_{n=1}^{N}-jw_n\left(\ddot{\mathbf{H}_0^n}\right)^H\mathbf{H}_0$, and $\boldsymbol{\zeta}_{6}=\sum_{n=1}^{N}w_n^2\left(\mathbf{H}_0\right)^H\mathbf{H}_0$. Among the estimated parameters, only the TS error $\Delta\tau_{t,r}$ is a random variable  with prior information and is independent of the position parameters $\mathbf{p}_o=\left(o_x,o_y\right)^{T}$. Consequently,  $\mathbf{J}_{B}\left(\boldsymbol{\chi}\right)=\begin{bmatrix}
			\boldsymbol{0}_{2\times2}&\boldsymbol{0}_{2\times1}\\
			\boldsymbol{0}_{1\times2}&J_B^{\Delta\tau_{t,r}\Delta\tau_{t,r}}
		\end{bmatrix}$  and $J_B^{\Delta\tau_{t,r}\Delta\tau_{t,r}}$ is derived as 
		\begin{equation}\label{79}
			J_B^{\Delta\tau_{t,r}\Delta\tau_{t,r}}=\frac{\partial \Delta \tau_{t,r}^{T}}{\partial \Delta \tau_{t,r}} \left(\sigma_{T}^{2}\right)^{-1} \frac{\partial \Delta \tau_{t,r}}{\partial \Delta \tau_{t,r}}=\frac{1}{\sigma_{T}^2}.
		\end{equation}
				\begin{figure*}[b] 
			\hrulefill 
						\begin{align}
				g^{'}\left(\hat{\psi}\right)=-\hat{\psi}e^{-\hat{\psi}^{2}/2}
				-\sqrt{2\pi}\left(Q(\hat{\psi})-\hat{\psi}\frac{1}{\sqrt{2\pi}}e^{-\hat{\psi}^{2}/2}\right)=-\sqrt{2\pi}Q\left(\hat{\psi}\right)\label{78}.~~~~~~~~~~~~~~
				\tag{80}\end{align}  
			\begin{align}
				\left|f_j\left(\mathbf{w}_c,\mathbf{w}_0\right)-\mu_j\right|^2
				&=\left|\operatorname{Tr}\left(\boldsymbol{\zeta}_{j}\mathbf{w}_c\mathbf{w}_c^{H}\right)+\operatorname{Tr}\left(\boldsymbol{\zeta}_{j}\mathbf{w}_0{\mathbf{w}_0}^H\right)-\mu_{j}\right|^2 \tag{86}\label{8511}\\
				&=\underbrace{\left|\mathbf{w}_c^{H}\boldsymbol{\zeta}_j\mathbf{w}_c\right|^2}_{\text{non-convex}}+\underbrace{\left|\Psi_{j}\right|^2}_{\text{non-convex}}-\underbrace{2\Re\left\{\left(
					\Psi_{j}\right)^{*}\mathbf{w}_c^{H}\boldsymbol{\zeta}_j\mathbf{w}_c\right\}}_{\text{non-convex}}\tag{87}\label{8611},~~~~~~~~~~~~~~~~
			\end{align}
			\vspace{-1em}
		\end{figure*}
		\vspace{-0.5cm}
		\section{proof of lemma 4.1}
		According to \cite{9685691}, for any $\psi,\hat{\psi}\in\mathbb{R}$, $Q$-function satisfies
		\begin{equation}\label{761}
			Q\left(\psi\right)\leq\hat{b}\left(\hat{\psi}\right)e^{-\hat{a}\left(\hat{\psi}\right)\psi}+\hat{
				c}\left(\hat{\psi}\right),
		\end{equation}
		where $\hat{a}\left(\hat{
			\psi}\right)\triangleq \text{max}\left\{\frac{e^{-\frac{\hat{\psi}^2}{2}}}{\sqrt{2\pi}Q\left(\hat{\psi}\right)}, \hat{\psi}\right\}$, $\hat{b}\left(\hat{\psi}\right)=\frac{1}{\sqrt{2\pi}\hat{a}\left(\hat{\psi}\right)}e^{\hat{a}\left(\hat{\psi}\right)\hat{\psi}-\frac{\hat{\psi}^2}{2}}$, and  $\hat{c}\left(\hat{\psi}\right)= Q\left(\hat{\psi}\right)-Q_1\left(\hat{\psi}\right)$, with $Q_1\left(\hat{\psi}\right)=\hat{b}\left(\hat{\psi}\right)e^{-\hat{a}\left(\hat{\psi}\right)\hat{\psi}}$.  To acquire the real value of $\hat{a}\left(\hat{\psi}\right)$, we have the following lemma.
		\begin{Lemma}\label{b.1}
			$\frac{e^{-\frac{\hat{\psi}^2}{2}}}{\sqrt{2\pi}Q\left(\hat{\psi}\right)}-\hat{\psi} > 0$ holds for any $\hat{\psi}\in\mathbb{R}$.
		\end{Lemma}
		\vspace{-0.1cm}
		\begin{proof}
			For facilitating the proof, defining $U\left(
			\hat{\psi}\right)$ as $U\left(
			\hat{\psi}\right)=e^{-\frac{\hat{\psi}^2}{2}}-\sqrt{2\pi}\hat{\psi} Q\left(\hat{\psi}\right)$. 
			The first-order derivative of $U\left(\hat{\psi}\right)$ with respect to $\hat{\psi}$ is given by (\ref{78}). 
			Since $Q(\hat{\psi})>0$ for any $\hat{\psi}\in\mathbb{R}$,  $U(\hat{\psi})$ is strictly decreasing in $\hat{\psi}$. Moreover, it follows form L'H\^opital's Rule that $\mathop{\text{lim}}\limits_{\hat{\psi}\rightarrow\infty}U\left(\hat{\psi}\right)=0$. Consequently, $U\left(\hat{\psi}\right)>0$ for all $\hat{\psi}\in\mathbb{R}$, which implies that $U\left(
			\hat{\psi}\right)=e^{-\frac{\hat{\psi}^2}{2}}-\sqrt{2\pi}\hat{\psi} Q\left(\hat{\psi}\right)>0$, i.e., 	$\frac{e^{-\frac{\hat{\psi}^2}{2}}}{\sqrt{2\pi}Q\left(\hat{\psi}\right)}-\hat{\psi} > 0$. 
		\end{proof} \hfill $\blacksquare$
		
		Based on Lemma \ref{b.1}, 
		we have  $\hat{a}\left(\hat{
			\psi}\right)= \frac{e^{-\frac{\hat{\psi}^2}{2}}}{\sqrt{2\pi}Q\left(\hat{\psi}\right)}$ and then $\hat{b}\left(\hat{
			\psi}\right)$, $\hat{c}\left(\hat{
			\psi}\right)$ are simplified as $\hat{b}\left(\hat{
			\psi}\right)=Q\left(\hat{\psi}\right)e^{\hat{a}\left(\hat{
				\psi}\right)\hat{	\psi}}$, $\hat{c}\left(\hat{
			\psi}\right)=0$, respectively.  Hence, (\ref{761}) is rewritten as 
		\setcounter {equation} {80}
		\begin{equation}
			Q\left(\psi\right)\leq 	Q\left(\hat{\psi}\right)e^{-\hat{a}\left(\hat{\psi}\right)\left(\psi-\hat{	\psi}\right)}\label{801}.
		\end{equation}
		By setting $\psi=\frac{\sqrt{L}\left(C\left(\gamma\right)-R_c\right)\text{ln}2}{B\left(\gamma\right)}$ and $\hat{\psi}=\frac{\sqrt{L}\left(C\left(\gamma\right)-R_c^{'}\right)\text{ln}2}{B\left(\gamma\right)}$,  the upper bound of $U\left(R_c\right)$  is derived as
		\begin{align}
			&\log_{10} \! \left( \frac{Q(\psi)}{R_c L} \right) 
			\le \log_{10} \! \left( \frac{Q(\hat{\psi}) e^{-\hat{a}(\hat{\psi})(\psi-\hat{\psi})}}{R_c L} \right) \label{811} \\
			&= \log_{10} Q(\hat{\psi}) - \log_{10}(R_c L) - \hat{a}(\hat{\psi})(\psi-\hat{\psi}) \log_{10}^e, \label{821}
		\end{align}
		where (\ref{811}) is obtained according to  (\ref{801}), and (\ref{821}) is obtained by expanding the logarithmic terms in (\ref{801}). Besides, the upper bound of the LHS of (\ref{36_11})  is obtained by taking the first-order Taylor expansion of the LHS of (\ref{36_11}).

		\section{proof of lemma 4.2}
		The first-order derivative of $\hat{\gamma}$ with respect to $\gamma$ is
		\begin{equation}\label{48}
			\frac{d\hat{\gamma}}{d\gamma}=\frac{\sqrt{L}\left(2\gamma+\gamma^{2}+ R_c\ln2-\log\left(\gamma+1\right)\right)}{\left(\gamma+1\right)^{3}\left(1-\frac{1}{\left(\gamma+1\right)^2}\right)^{\frac{3}{2}}}.
		\end{equation}
		To analyze the monotonicity of $\hat{\gamma}$ with respect to $\gamma$, let $E$ denote the numerator of (\ref{48}). The first-order derivative of $E$ with respect to $\gamma$ is given as  
		\begin{equation}\label{49}
			\frac{dE}{d\gamma}=\sqrt{L}\left(2\gamma+2-\frac{1}{\gamma+1}\right).
		\end{equation}
		The function in  (\ref{49}) monotonically  increases with respect to $\gamma$, due to  that the first-order derivative of (\ref{49}) is $\sqrt{L}\left(2+{\frac{1}{\left(\gamma +1\right)^{2}}}\right)\textgreater 0$. Thus, $\frac{dE}{d\gamma}$ is strictly increasing with respect to $\gamma$, and since $\left.\frac{dE}{d\gamma}\right|_{\gamma=0}=\sqrt{L}>0$, we conclude that $E$ monotonically increases over $\gamma$.  Furthermore, for all $\gamma\textgreater0$,  $E$ remains positive,  which implies $\hat{\gamma}$ monotonically increases with respect to $\gamma$. It is worth noting that  Q-function is a monotonically decreasing function according to its  fundamental definition. Combining these  analytical results,  we can rigorously establish that the BER $\rho_{b}$ in (\ref{46}) demonstrates a monotonic decreasing behavior  as $\gamma$ increases.    

		\section{proof of lemma 4.3}
		\vspace{-0.1cm}
		The expansion of $\left|f_j\left(\mathbf{w}_c,\mathbf{w}_0\right)-\mu_j\right|^2$ is expressed as (\ref{8611}), where (\ref{8511}) is obtained  by substituting the definition of $f_j\left(\mathbf{w}_c,\mathbf{w}_0\right)$ into the original expression; (\ref{8611}) is obtained by expanding (\ref{8511}) and leveraging the property that $\operatorname{Tr}\left(\mathbf{A}\mathbf{B}\right)=\operatorname{Tr}\left(\mathbf{B}\mathbf{A}\right)$,
		where $\Psi_j=\mathbf{w}_0^H\boldsymbol{\zeta}_{j}{\mathbf{w}_0}-\mu_{j}$. The first non-convex term of (\ref{8611}) is equivalently transformed into  $\left|\mathbf{w}_c^{H}\boldsymbol{\zeta}_j\mathbf{w}_c\right|^2=\hat{\mathbf{w}}_c^{H}\mathbf{G}_j\hat{\mathbf{w}}_c$,  where $\hat{\mathbf{w}}_c=\text{vec}\left(\mathbf{w}_c\mathbf{w}_c^{H}\right)$ and $\mathbf{G}_j=\left(\boldsymbol{\zeta}_j\right)^{*}\otimes\boldsymbol{\zeta}_j$, by utilizing the property  that $\operatorname{Tr}\left(\mathbf{A}\mathbf{B}\mathbf{C}\mathbf{D}^{H}\right)=\left(\text{vec}\left(\mathbf{D}\right)\right)^{H}\left(\mathbf{C}\otimes\mathbf{A}\right)\text{vec}\left(\mathbf{B}\right)$. Then, performing the first-order Taylor expansion at $\hat{\mathbf{w}}_c^{'}=\text{vec}\left(\mathbf{w}_c^{'}\left(\mathbf{w}_c^{'}\right)^{H}\right)$,  $\hat{\mathbf{w}}_c^{H}\mathbf{G}_j\hat{\mathbf{w}}_c$ is expressed as (\ref{85}), 
		where
		$2\Re\left\{\left(\hat{\mathbf{w}}_c^{'}\right)^{H}\mathbf{G}_j\hat{\mathbf{w}}_c\right\}$ is reformulated as (\ref{86}). Here, equality $(a)$ is obtained by substituting the definition of $\mathbf{w}_c^{'}$  and leveraging the property   $\operatorname{Tr}\left(\mathbf{A}\mathbf{B}\mathbf{C}\mathbf{D}^{H}\right)=\left(\text{vec}\left(\mathbf{D}\right)\right)^{H}\left(\mathbf{C}\otimes\mathbf{A}\right)\text{vec}\left(\mathbf{B}\right)$; equality $(b)$ is obtained by exploiting the Hermitian property of $\mathcal{G}_j^{'}=\left(\boldsymbol{\zeta}_j\right)^{H}\mathbf{w}_c^{'}\left(\mathbf{w}_c^{'}\right)^{H}\boldsymbol{\zeta}_j$.
		To facilitate solution, the third term of (\ref{8611}) is equivalently transformed into the following form, i.e.,
		\setcounter{equation}{89}
		\begin{align}
			&2\Re\left\{\left(
			\Psi_{j}\right)^{*}\mathbf{w}_c^{H}\boldsymbol{\zeta}_j\mathbf{w}_c\right\}\label{841}\\
			&=2\Re\left\{\left({\mathbf{w}_0}^H\boldsymbol{\zeta}_{j}\mathbf{w}_0\right)^{*}\cdot\left({\mathbf{w}_c}^H\boldsymbol{\zeta}_{j}\mathbf{w}_c\right)\right\}\nonumber\\
			&~~~-2\Re\left\{\mu_{j}^*\left({\mathbf{w}_c}^H\boldsymbol{\zeta}_{j}\mathbf{w}_c\right)\right\}\label{851}\\
			&=\left(g^*h+gh^{*}\right)-\left(\mu_{j}^{*}h+\mu_{j}h^*\right)\label{861},
		\end{align} 
		where (\ref{851}) is obtained by substituting the definition of $\Psi_j$  into (\ref{841}), and (\ref{861}) is obtained by leveraging the property  $A+A^{*}=2\Re\left\{A\right\}$, with $g=\mathbf{w}_0^H\boldsymbol{\zeta}_{j}\mathbf{w}_0$ and $h=\mathbf{w}_c^H\boldsymbol{\zeta}_{j}\mathbf{w}_c$, respectively. Performing the first-order Taylor approximation at $g^{'}$ and $h^{'}$ for the first term of (\ref{861}), we have,
		\begin{equation}\label{88}
			\begin{aligned}
				g^*h+gh^{*}=2\Re\left\{\mathbf{a}_1^{'}\mathbf{w}_0\right\}+2\Re\left\{\mathbf{a}_2^{'}\mathbf{w}_c\right\}	+C,
			\end{aligned}
		\end{equation} 
		where $\mathbf{a}_1^{'}=\left(\left(h^{'}\right)^{*}\left(\boldsymbol{\zeta}_j\right)\mathbf{w}_0^{'}+h^{'}\left(\boldsymbol{\zeta}_j\right)^H\mathbf{w}_0^{'}\right)^H$ and $\mathbf{a}_2^{'}=\left(\left(g^{'}\right)^{*}\left(\boldsymbol{\zeta}_j\right)\mathbf{w}_c^{'}+g^{'}\left(\boldsymbol{\zeta}_j\right)^H\mathbf{w}_c^{'}\right)^H$, respectively,
		with $h^{'}=\left(\mathbf{w}_c^{'}\right)^H\boldsymbol{\zeta}_{j}\mathbf{w}_c^{'}$ and $g^{'}=\left(\mathbf{w}_0^{'}\right)^H\boldsymbol{\zeta}_{j}\mathbf{w}_0^{'}$. Furthermore, $C$ is a constant, which is derived as $C=\left(g^{'}\right)^{*}h^{'}+g^{'}\left({h^{'}}\right)^*-2\Re\left\{\mathbf{a}_1^{'}\mathbf{w}_0^{'}\right\}-2\Re\left\{\mathbf{a}_2^{'}\mathbf{w}_c^{'}\right\}$.
		For the second term of (\ref{861}), it is transformed into the following form, i.e.,
		\begin{align}
			\mu_{j}^{*}h+\mu_{j}h^*
			&=\left(
			\mu_{j}\right)^{*}\mathbf{w}_c^{H}\mathbf{\zeta}_j\mathbf{w}_c
			+\left(
			\mu_{j}\right)\left(\mathbf{w}_c^{H}\boldsymbol{\zeta}_j\mathbf{w}_c\right)^{*}\label{931}\\
			&=\mathbf{w}_c^{H}\underbrace{
				\left[\left(\mu_{j}\right)^{*}\boldsymbol{\zeta}_j+\mu_{j}\left(\boldsymbol{\zeta}_j\right)^{H}\right]}_{\mathbf{\Lambda}_j}\mathbf{w}_c,\label{90}
		\end{align}
		where (\ref{931}) is obtained by substituting the definition of $h$ into the original expression and (\ref{90}) is obtained by rearranging terms in (\ref{931}). Similar to (\ref{85}),   $\mathbf{w}_c^{H}\mathbf{\Lambda}_j\mathbf{w}_c$ is reformulated as the convex form given in (\ref{91}) via SCA.
		\begin{figure*}[t] 
			\vspace{-2em}  
			\begin{equation}\label{85}
				\begin{aligned}
					~~~~~~~\hat{\mathbf{w}}_c^{H}\mathbf{G}_j\hat{\mathbf{w}}_c= \left(\hat{\mathbf{w}}_c^{'}\right)^H\mathbf{G}_j\hat{\mathbf{w}}_c^{'}-2\Re\left\{\left(\hat{\mathbf{w}}_c^{'}\right)^{H}\mathbf{G}_j\hat{\mathbf{w}}_c^{'}\right\}
					+2\Re\left\{\left(\hat{\mathbf{w}}_c^{'}\right)^{H}\mathbf{G}_j\hat{\mathbf{w}}_c\right\},~~~~~~~~~~~~~~~~~~~~~~~~~~~~~~~~~~~~
				\end{aligned}\tag{88}
			\end{equation}
			\begin{equation}\label{86}
				\begin{aligned}
					2\Re\left\{\left(\hat{\mathbf{w}}_c^{'}\right)^{H}\mathbf{G}_j\hat{\mathbf{w}}_c\right\}&\stackrel{(a)}{=}2\Re\left\{\mathbf{w}_c^{H}\mathcal{G}_j^{'}\mathbf{w}_c\right\}
					\stackrel{(b)}{=}2\left|\left(\mathbf{w}_c^{'}\right)^{H}\left(\boldsymbol{\zeta}_j\right)^{H}\mathbf{w}_c \right|^2.~~~~~~~~~~~~~~~~~~~~~~~~~~~~~~~~~~~~~~~~~
				\end{aligned}\tag{89}
			\end{equation}
			\begin{equation}\label{91}
				\begin{aligned}
					\mathbf{w}_c^{H}\mathbf{\Lambda}_j\mathbf{w}_c&=2\mathcal{R}\left\{\left(\mathbf{w}_c^{'}\right)^H\mathbf{\Lambda}_j\mathbf{w}_c\right\}-2\mathcal{R}\left\{\left(\mathbf{w}_c^{'}\right)^H\mathbf{\Lambda}_j\mathbf{w}_c^{'}\right\}
					+\left(\mathbf{w}_c^{'}\right)^H\mathbf{\Lambda}_j\mathbf{w}_c^{'}.~~~~~~~~~~~~~~~~~~~~~~~~~~~~~
				\end{aligned}
			\end{equation}
			\begin{equation}\label{93}
				\left|\Psi_j\right|^2= 2\left|\left(\mathbf{w}_0^{\left(i-1\right)}\right)^{H}\left(\boldsymbol{\zeta}_j\right)^{H}\mathbf{w}_0 \right|^2-2\mathcal{R}\left\{\left(\mathbf{w}_0^{\left(i-1\right)}\right)^H\mathbf{\Lambda}_j\mathbf{w}_0\right\}\tag{99},~~~~~~~~~~~~~~~~~~~~~~~~~~~~~~~~~~~~~~~~~~~~~~~
			\end{equation}
			\begin{equation}\label{96}
				\begin{aligned}
					\left|f_j\left(\mathbf{w}_c,\mathbf{w}_0\right)-\mu_j\right|^2
					&= 2\left|\left(\mathbf{w}_c^{'}\right)^{H}\left(\boldsymbol{\zeta}_j\right)^{H}\mathbf{w}_c \right|^2+2\left|\left(\mathbf{w}_0^{'}\right)^{H}\left(\boldsymbol{\zeta}_j\right)^{H}\mathbf{w}_0 \right|^2
					+2\mathcal{R}\left\{\left(\mathbf{a}_1^{'}-\left(\mathbf{w}_0^{'}\right)^H\mathbf{\Lambda}_j\right)\mathbf{w}_0\right\}\\
					&~~~+2\mathcal{R}\left\{\left(\mathbf{a}_2^{'}-\left(\mathbf{w}_c^{'}\right)^H\mathbf{\Lambda}_j\right)\mathbf{w}_c\right\}.
				\end{aligned}\tag{100}
			\end{equation}
			\hrulefill 
		\end{figure*}
		Then, the second term of (\ref{8611}) is given as
		\begin{align}
			\left|\Psi_j\right|^2&=\left|\mathbf{w}_0\boldsymbol{\zeta}_{j}{\mathbf{w}_0}^H-\mu_{j}\right|^2,\label{961}\\
			&=\underbrace{\left|\mathbf{w}_0^{H}\boldsymbol{\zeta}_j\mathbf{w}_0\right|^2}_{\text{non-convex}}+\left|\mu_{j}\right|^2-\underbrace{2\Re\left\{\left(
				\mu_{j}\right)^{*}\mathbf{w}_0^{H}\boldsymbol{\zeta}_j\mathbf{w}_0\right\}}_{\text{non-convex}},\label{92}
		\end{align}
		where (\ref{961}) is obtained by substituting the definition of $\Psi_j$ and (\ref{92}) is obtained by expanding (\ref{961}). Since (\ref{92}) has a similar structure  to (\ref{8611}), its convex reformulation is derived in the same manner. Specifically,  $\left|\Psi_j\right|^2$ is reformulated as the convex form given in (\ref{93}).
		Notably, the constant term is ignored in the  process of deriving (\ref{93}) due to the constant term does not affect the solution of the problem. 
		Based on (\ref{86}), (\ref{88}), (\ref{91}) and (\ref{93}), $\left|f_j\left(\mathbf{w}_c,\mathbf{w}_0\right)-\mu_j\right|^2$ is reformulated into the convex version given in (\ref{96}). 
		\bibliographystyle{IEEEtran}
		\bibliography{reference11}
	\end{document}